\PassOptionsToPackage{noend}{algorithmic}

\documentclass[a4paper,UKenglish,cleveref, autoref, thm-restate,notab, nosubfigcap,nolineno]{socg-lipics-v2021}
\hideLIPIcs

\usepackage{amsmath}
\usepackage{mathtools}
\usepackage{tikz}
\usetikzlibrary{matrix,arrows,cd,calc}
\usepackage{mathbbol, graphicx, float}
\usepackage{algorithm}
\usepackage{algpseudocode}
\algrenewcommand\algorithmicrequire{\textbf{Input:}}
\algrenewcommand\algorithmicensure{\textbf{Output:}}
\usepackage{setspace}
\usepackage{enumitem}

\usepackage{subcaption}
\usepackage{graphicx}
\graphicspath{{figures/}}
\usepackage{amsmath}
\usepackage{tikz}

\title{Toroidal Coordinates:  Decorrelating Circular Coordinates With Lattice Reduction}
\titlerunning{Toroidal Coordinates:  Decorrelating Circular Coordinates With Lattice Reduction}

\author{Luis Scoccola}{Department of Mathematics, Northeastern University, USA \and \url{https://luisscoccola.com}}{l.scoccola@northeastern.edu}{https://orcid.org/0000-0002-4862-722X}{supported by the NSF through
grants CCF-2006661 and CAREER award DMS-1943758}

\author{Hitesh Gakhar}{Department of Mathematics, The University of Oklahoma, USA \and \url{https://hiteshgakhar.com}}{hiteshgakhar@ou.edu}{https://orcid.org/0000-0001-7728-6738}{}

\author{Johnathan Bush}{Department of Mathematics, University of Florida, USA \and \url{https://people.clas.ufl.edu/bush-j/}}{bush.j@ufl.edu}{https://orcid.org/0000-0002-6404-8324}{supported by the NSF-Simons Southeast Center for Mathematics and Biology through NSF grant
DMS-1764406 and Simons Foundation grant 594594}

\author{Nikolas Schonsheck}{Department of Mathematical Sciences, University of Delaware, USA \and \url{https://niko-schonsheck.github.io/}}{nischon@udel.edu}{https://orcid.org/0000-0002-6177-4865}{supported by the Air Force Office
of Scientific Research through award number FA9550-21-1-0266}

\author{Tatum Rask}{Department of Mathematics, Colorado State University, USA \and \url{https://sites.google.com/view/tatumrask/} }{tatum.rask@colostate.edu}{}{}

\author{Ling Zhou}{Department of Mathematics, The Ohio State University, USA \and \url{https://www.ling-zhou.com/} }{zhou.2568@osu.edu}{https://orcid.org/0000-0001-6655-5162}{}

\author{Jose A. Perea}{Department of Mathematics and Khoury College of Computer Sciences, Northeastern University, USA \and \url{https://www.joperea.com/} }{j.pereabenitez@northeastern.edu}{https://orcid.org/0000-0002-6440-5096}{supported by the NSF through grants CCF-2006661 and CAREER award DMS-1943758}

\authorrunning{L. Scoccola, H. Gakhar, J. Bush, N. Schonsheck, T. Rask, L. Zhou, and J. A. Perea} 

\Copyright{Luis Scoccola, Hitesh Gakhar, Johnathan Bush, Nikolas Schonsheck, Tatum Rask, Ling Zhou, and Jose A. Perea} 



\category{} 

\supplement{Proof-of-concept implementation at \url{https://github.com/LuisScoccola/DREiMac} \cite{implementation-2}}

\funding{This material is based upon work initiated during the 2022 Mathematics Research Community
\emph{Data Science at the Crossroads of Analysis, Geometry, and Topology}
supported by the National Science Foundation under Grant Number DMS 1916439.}

\EventEditors{Erin W. Chambers and Joachim Gudmundsson}
\EventNoEds{2}
\EventLongTitle{39th International Symposium on Computational Geometry
(SoCG 2023)}
\EventShortTitle{SoCG 2023}
\EventAcronym{SoCG}
\EventYear{2023}
\EventDate{June 12--15, 2023}
\EventLocation{Dallas, Texas, USA}
\EventLogo{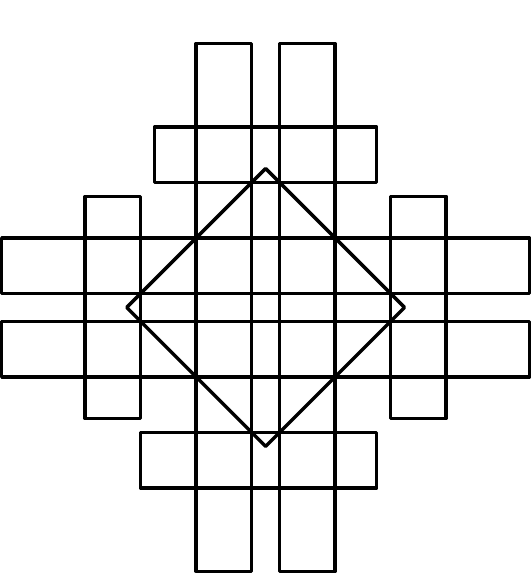}
\SeriesVolume{258}
\ArticleNo{XX}

\def\noteson{\gdef\todo##1{\noindent{\color{blue}[to do: ##1]}}}

\noteson

\def\thoughtson{\gdef\thoughts##1{\noindent{\color{red}[thoughts: ##1]}}}

\thoughtson

\theoremstyle{definition}
\newtheorem{problem}[theorem]{Problem}

\newtheorem{construction}[theorem]{Construction}
\newtheorem{pipeline}[theorem]{Pipeline}

\DeclareMathAlphabet{\mathpzc}{OT1}{pzc}{m}{it}

\newcommand{\define}[1]{\emph{#1}}

\DeclareMathOperator*{\argmin}{argmin}

\newcommand{\estdirichlet}{\widehat{D}}
\newcommand{\innerprod}{{ \langle\text{-} \,, \text{-}\rangle}}
\newcommand{\SMV}{{\mathsf{dSMV}}}
\newcommand{\harmrep}{\mathtt{harmonicRepresentative}}
\newcommand{\lowenergyreps}{\mathtt{lowEnergyRepresentatives}}
\renewcommand{\lll}{\mathtt{LLL}}

\newcommand{\circcoords}{\mathtt{cc}}
\newcommand{\sparsecirccoords}{\mathtt{scc}}
\newcommand{\torcoords}{\mathtt{tc}}
\newcommand{\sparsetorcoords}{\mathtt{stc}}
\newcommand{\integrate}{\mathtt{integrate}}
\newcommand{\sparseintegrate}{\mathtt{sparseIntegrate}}

\newcommand{\tangent}{T}
\newcommand{\trace}{\mathsf{Tr}}

\renewcommand{\Im}{\mathsf{Im}}
\newcommand{\proj}{\mathsf{proj}}

\newcommand{\rips}{\mathsf{VR}}

\newcommand{\harco}{\mathcal{H}^1}

\newcommand{\simphom}{\mathsf{H}^1}

\newcommand{\simpco}{\mathsf{Z}^1}

\newcommand{\simpchzero}{\mathsf{C}^0}

\newcommand{\torus}{\mathbb{T}}
\renewcommand{\circle}{\mathbb{S}^1}

\newcommand{\dsf}{\mathsf{d}}

\newcommand{\Mcal}{\mathcal{M}}
\newcommand{\Ncal}{\mathcal{N}}
\newcommand{\Nbb}{\mathbb{N}}

\newcommand{\Qbb}{\mathbb{Q}}
\newcommand{\Rbb}{\mathbb{R}}

\newcommand{\Ucal}{\mathcal{U}}

\newcommand{\Zbb}{\mathbb{Z}}

\renewcommand{\epsilon}{\varepsilon}
\renewcommand{\phi}{\varphi}

\renewcommand{\to}{\xrightarrow{\;\;\;\;}}

\begin{document}

%
%
%
%
%

\maketitle

\begin{abstract}
    The circular coordinates algorithm of de Silva, Morozov, and Vejdemo-Johansson takes as input a dataset together with a cohomology class representing a $1$-dimensional hole in the data; the output is a map from the data into the circle that captures this hole, and that is of minimum energy in a suitable sense.
    However, when  applied to several cohomology classes, the output circle-valued maps can be ``geometrically correlated'' even if the chosen cohomology classes are linearly independent.
    It is shown in the original work that less correlated maps can be obtained with suitable integer linear combinations of the cohomology classes, with the linear combinations being chosen by inspection.
    In this paper, we identify a formal notion of geometric correlation between circle-valued maps which, in the Riemannian manifold case, corresponds to the Dirichlet form, a bilinear form derived from the Dirichlet energy.
    We describe a systematic procedure for constructing low energy torus-valued maps on data, starting from a set of linearly independent cohomology classes.
    We showcase our procedure with computational examples.
    Our main algorithm is based on the Lenstra--Lenstra--Lovász algorithm from computational number theory.
\end{abstract}

\section{Introduction}
\label{section:introduction}

\subparagraph{Motivation and problem statement}
Given a point cloud $X \subseteq \Rbb^n$ concentrated around a $k$-dimensional linear subspace, linear dimensionality reduction algorithms such as Principal Component Analysis are effective at finding a low-dimensional representation $X \to \Rbb^k$ of the data that preserves the linear structure.
The problem of finding low-dimensional representations of non-linear data is more involved; one reason being that it is often hard to make principled assumptions about which particular non-linear shape the data may have.
Topological Data Analysis provides tools allowing for the extraction of qualitative and quantitative topological information from discrete data.
These tools include persistent cohomology, which can be used, in particular, to identify circular features.

Given a dataset $X$ and a class $\alpha$ in the first integral persistent cohomology group of $X$, the circular coordinates algorithm of \cite{silva-vejdemo,silva-morozov-vejdemo} constructs a circle-valued representation $\circcoords_\alpha : X \to \circle$, which preserves the cohomology class $\alpha$ in a precise sense \cite[Theorem~3.2]{perea2020sparse}.
The circular coordinates algorithm is thus a principled non-linear dimensionality reduction algorithm, and has found various applications \cite{luo-kim-patania-vejdemo,wang-summa-pascucci-vejdemo}, particularly in neuroscience \cite{kang-xu-morozov,gardner-et-al,rybakken-baas-dunn}.

As observed in \cite[Section~3.9]{silva-morozov-vejdemo}, and reproduced in \cref{fig:gen2}, when several 
cohomology classes $\alpha_1, \dots, \alpha_k$ are used to produce a single torus-valued representation $(\circcoords_{\alpha_1}, \dots, \circcoords_{\alpha_k}) : X \to \circle \times \dots \times \circle = \torus^k$, this representation is often not the most natural. Indeed, even when the cohomology classes $\alpha_i$ are linearly independent (l.i.), the maps $\circcoords_{\alpha_i}$ can be ``geometrically correlated.'' Certain integer linear combinations of the cohomology classes, however, can yield decorrelated representations.
The problems of defining an appropriate notion of geometric correlation between circle-valued maps, and of using this notion to systematically decorrelate sets of circle-valued maps are left open in \cite{silva-morozov-vejdemo}.
In this paper, we address these two problems.

\begin{figure}[!htb]
\centering
\begin{tikzpicture}
  \node[inner sep=0pt,outer sep=0pt] at (0,0) (a){\includegraphics[width=.2\linewidth]{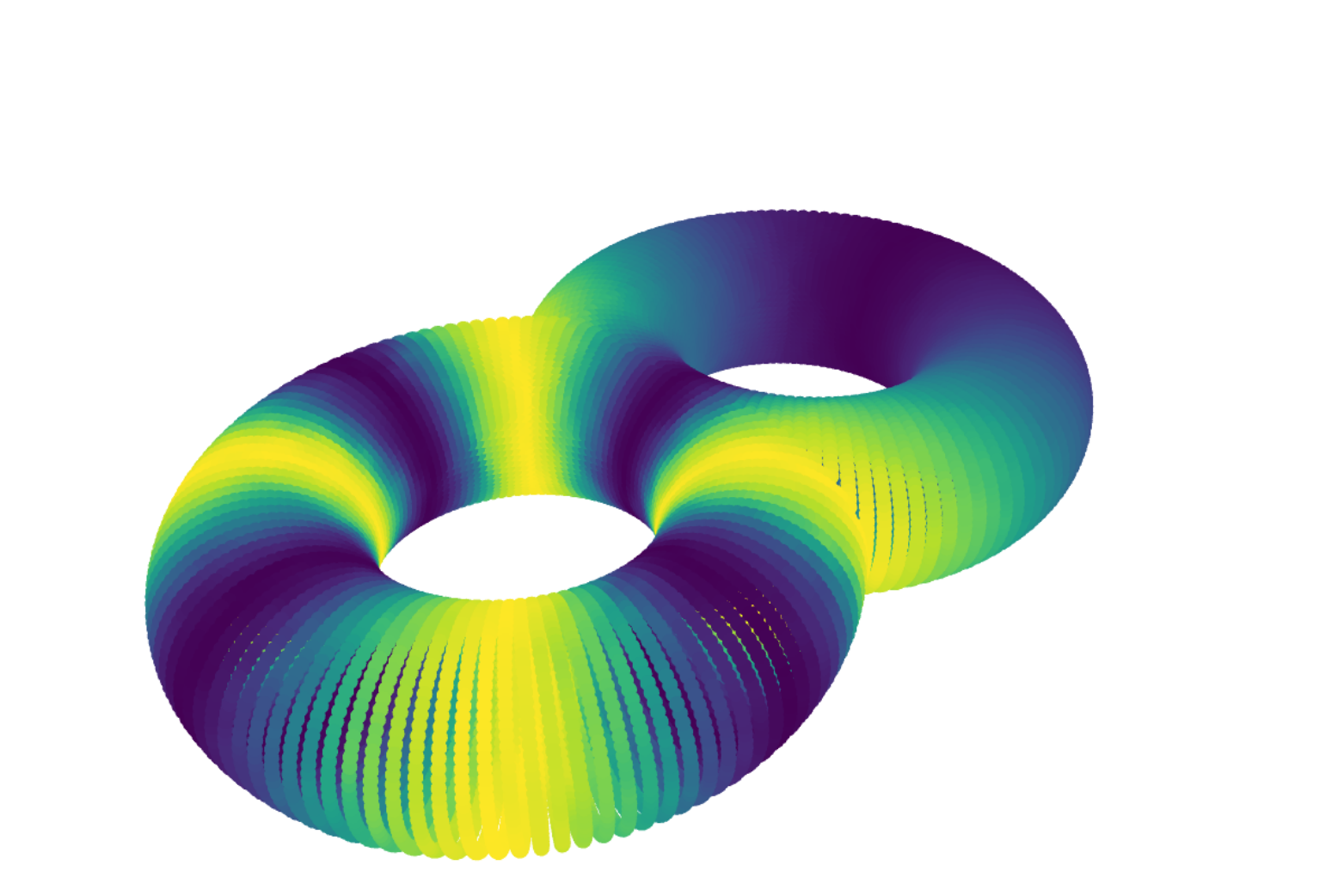}};
  \node[inner sep=0pt,outer sep=0pt] at (4,0) (b){\includegraphics[width=.133\linewidth]{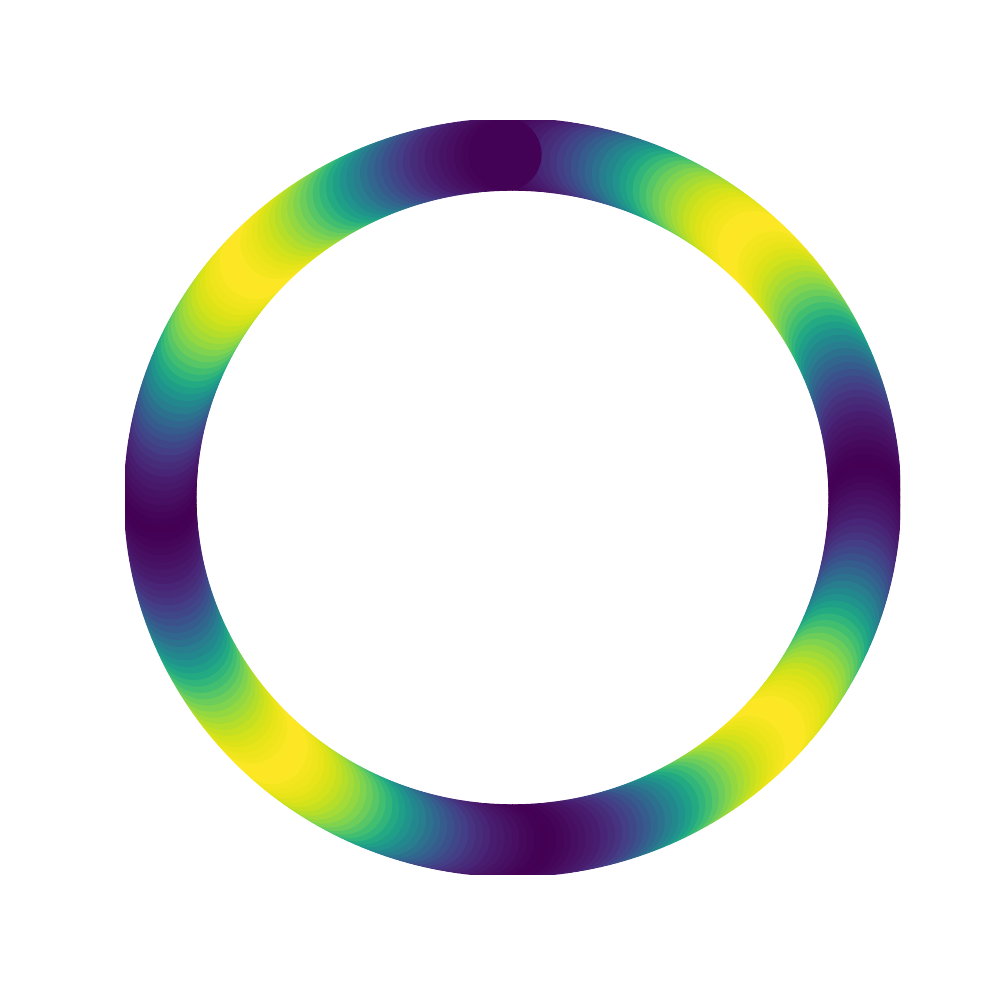}};
  \draw[line width=1pt,-stealth] ([xshift=2mm]a.east) -- ([xshift=-2mm]b.west)node[midway,above,text=black]{};
\end{tikzpicture}
\hfill
\begin{tikzpicture}
  \node[inner sep=0pt,outer sep=0pt] at (0,0) (a){\includegraphics[width=.2\linewidth]{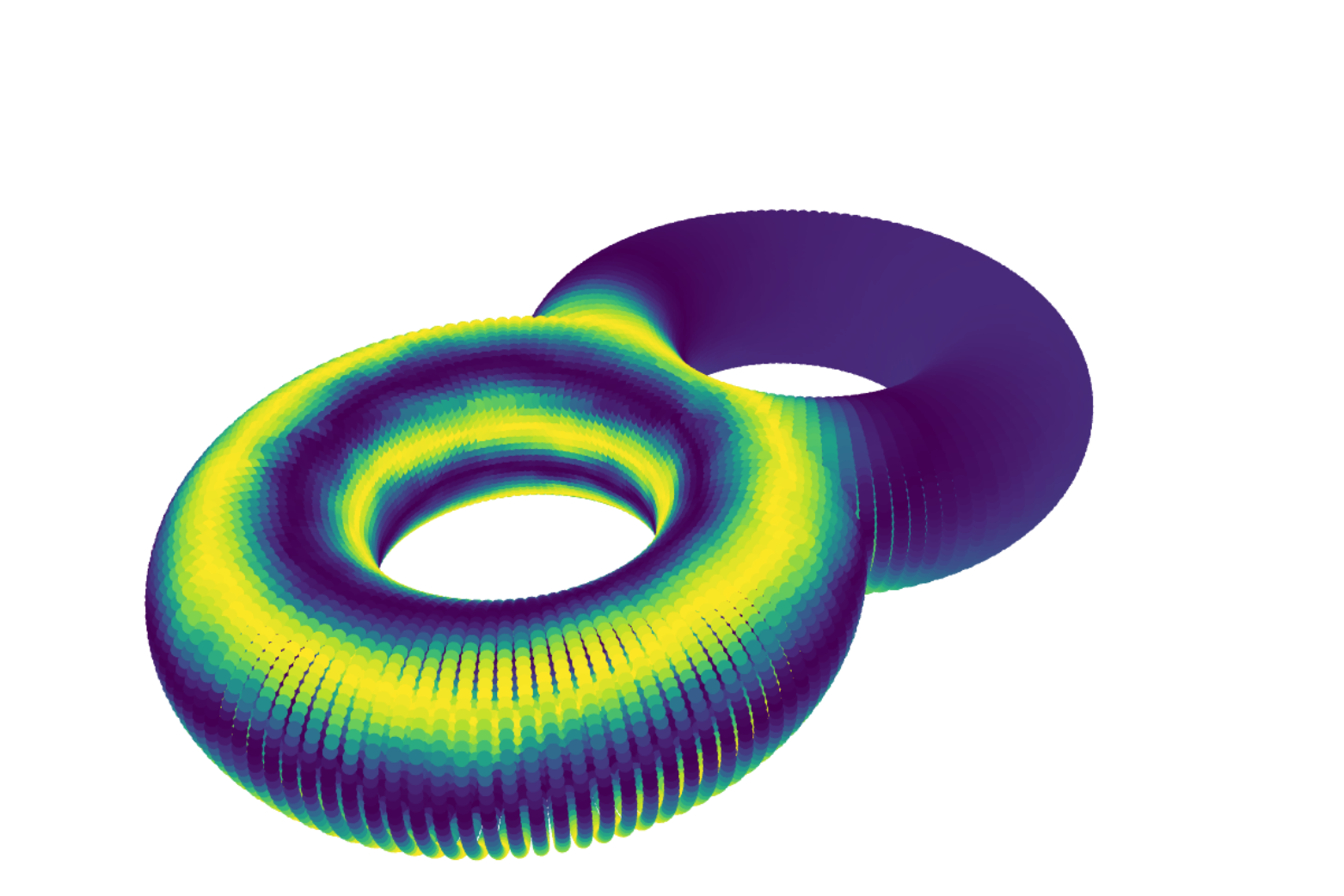}};
  \node[inner sep=0pt,outer sep=0pt] at (4,0) (b){\includegraphics[width=.133\linewidth]{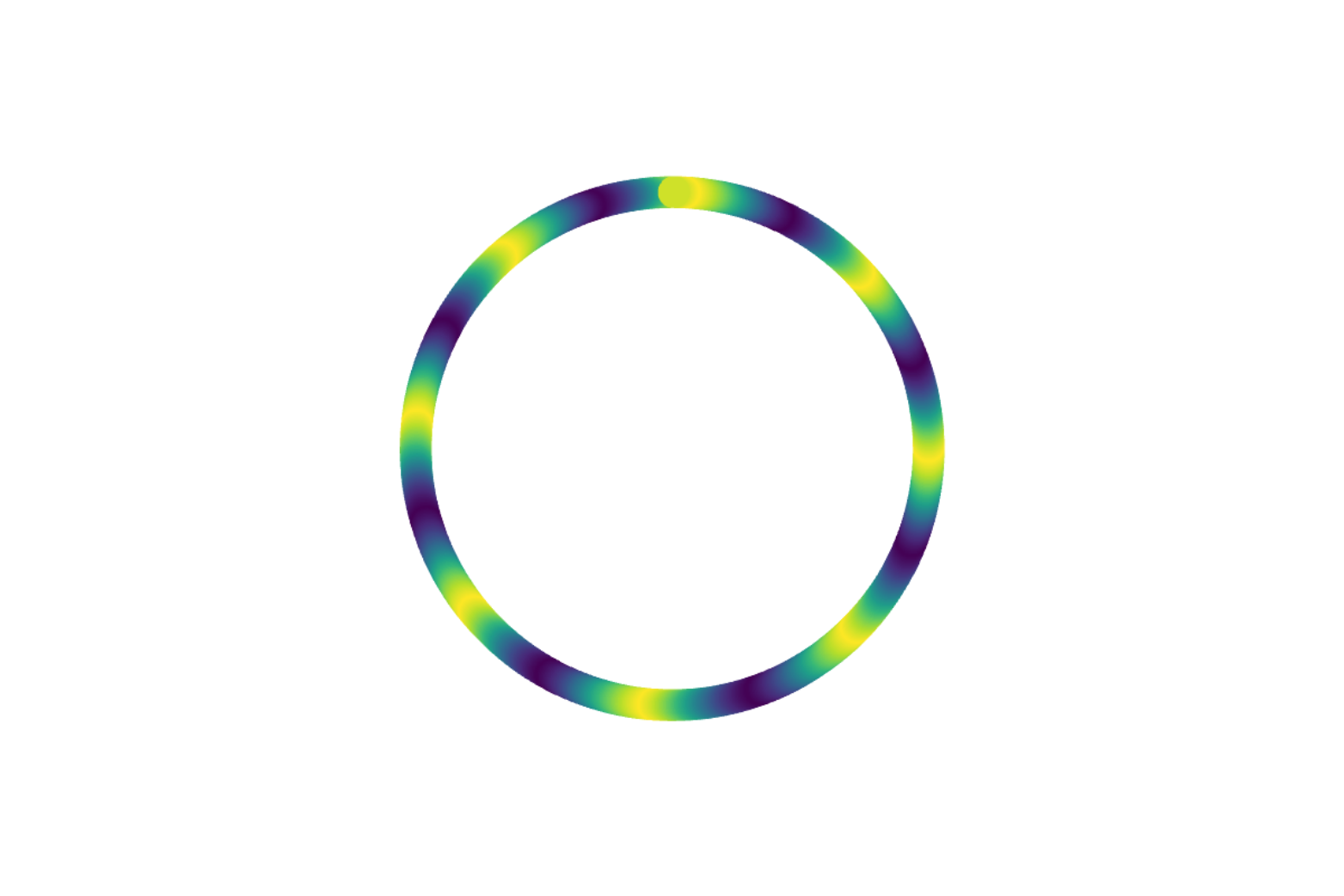}};
  \draw[line width=1pt,-stealth] ([xshift=2mm]a.east) -- ([xshift=-2mm]b.west)node[midway,above,text=black] {};
\end{tikzpicture}

\caption{An illustration of how we use colors to display circular coordinates on data.
We first color the circle $\circle$ with a smooth transition between yellow and violet, repeated four times; then, given a function into the circle, we color its domain by pulling back the coloring.
Depicted are the colorings induced on a genus two surface by the map that ``goes around a longitude'' (Left) and by the map that ``goes around a meridian'' (Right).}
\label{figure:example-coloring}
\vspace{0.8cm}

\begin{tikzpicture}
  \node[inner sep=0pt,outer sep=0pt, align=center] (a) at (0,0){
    \footnotesize
    Circular coordinates\\
    \,\\
  \includegraphics[width=0.09\linewidth]{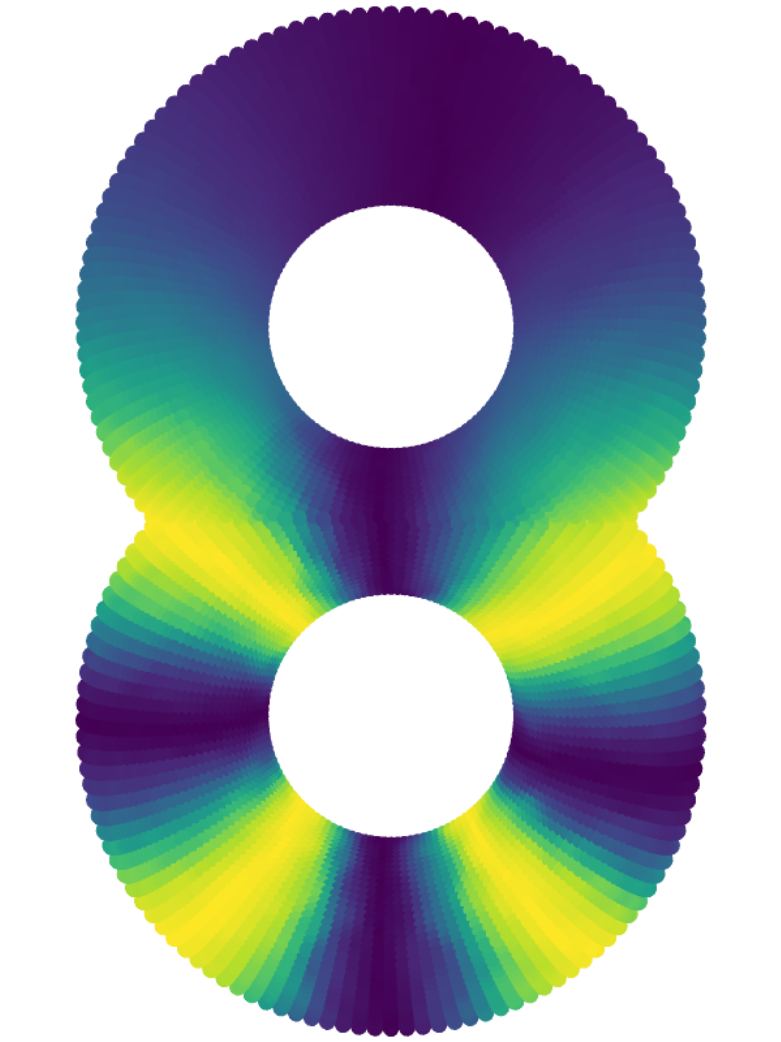}
  \includegraphics[width=0.09\linewidth]{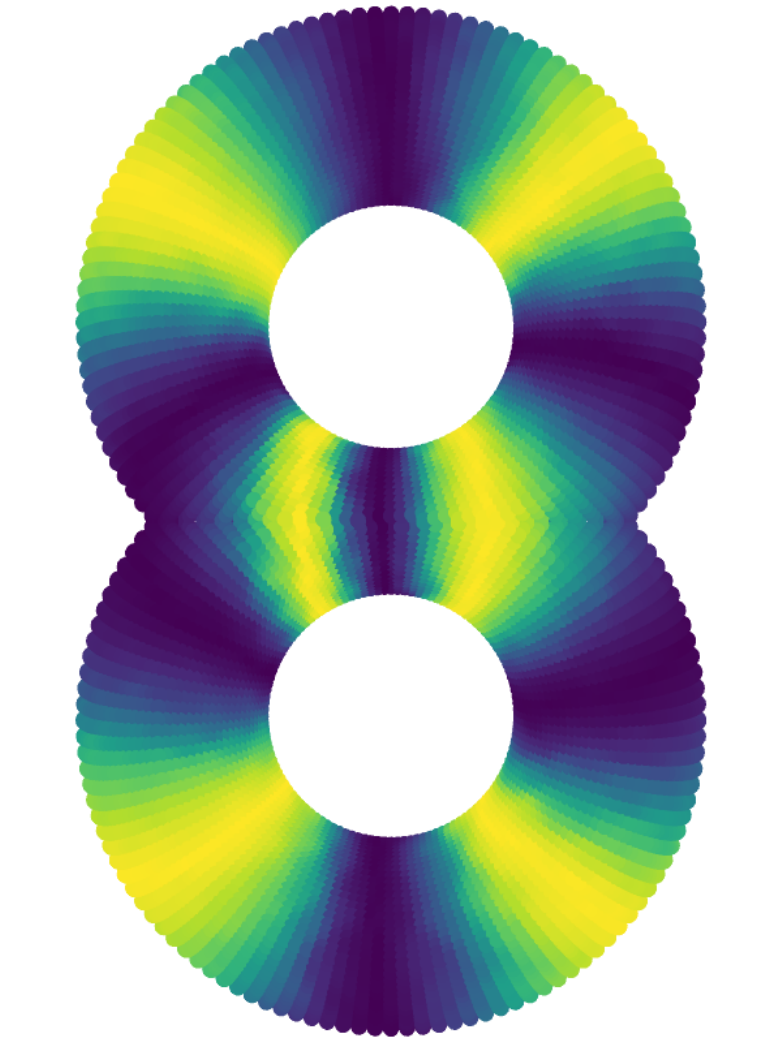}
  \includegraphics[width=0.09\linewidth]{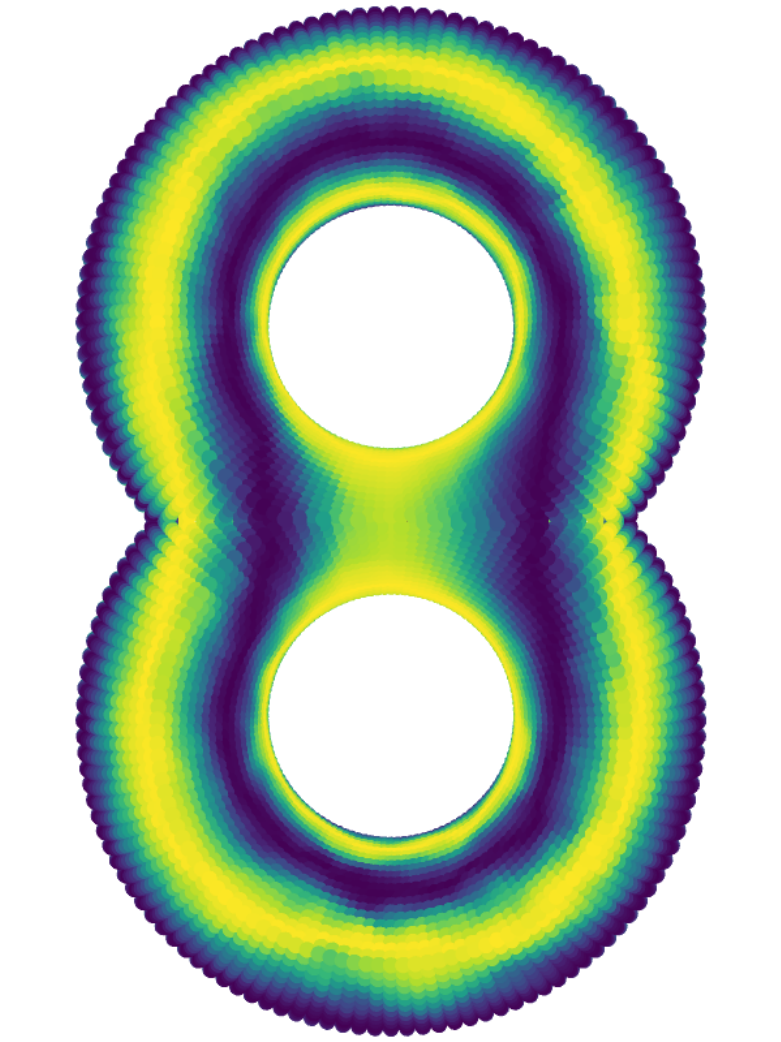}
  \includegraphics[width=0.09\linewidth]{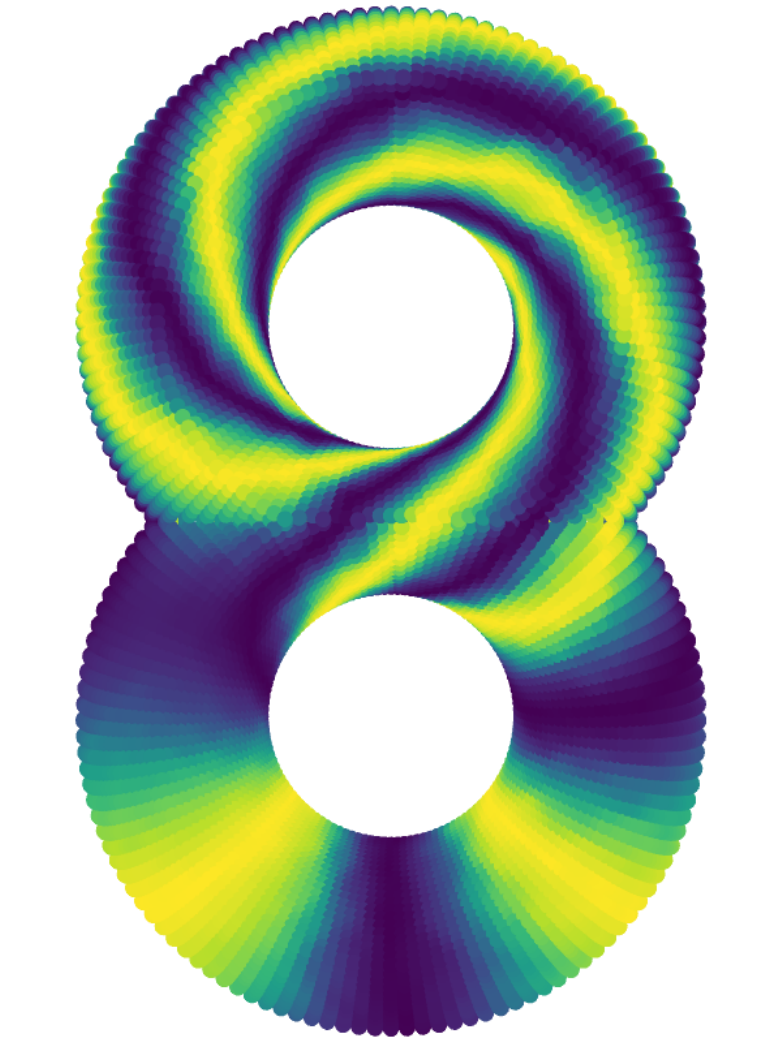}
  };
  \node[inner sep=0pt,outer sep=0pt, align=center] (b) at (8.5,0){
    \footnotesize
    Toroidal coordinates\\
    \,\\
  \includegraphics[width=0.09\linewidth]{figures/other_possibility_genus2/gen2_2.pdf}
  \includegraphics[width=0.09\linewidth]{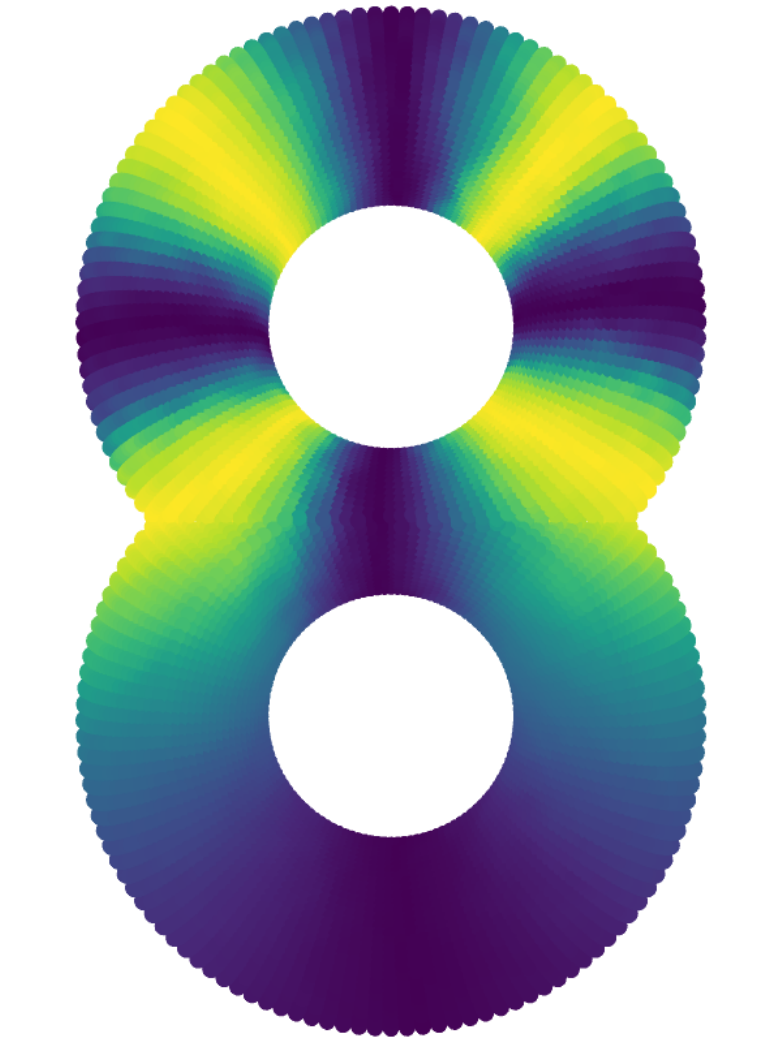}
  \includegraphics[width=0.09\linewidth]{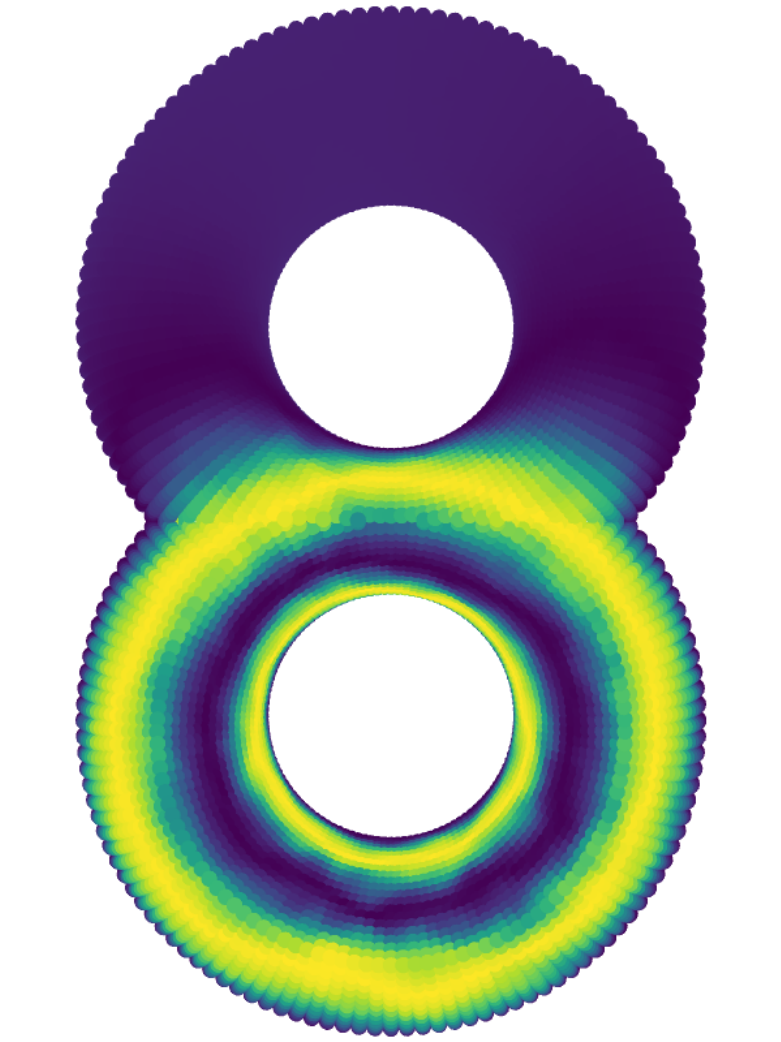}
  \includegraphics[width=0.09\linewidth]{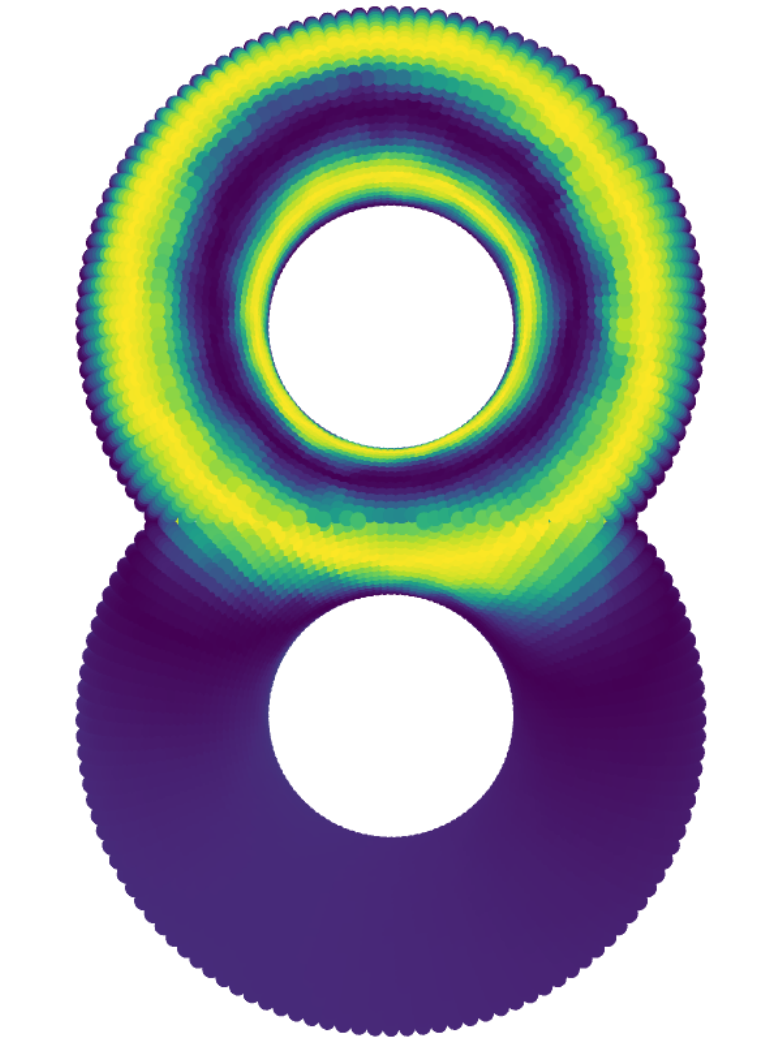}
  };
\draw[line width=1pt,-stealth] ([xshift=2mm,yshift=-3mm]a.east) -- ([xshift=-2mm,yshift=-3mm]b.west)node[midway,above,text=black]
{\scriptsize
    $\begin{pmatrix}
    1 &  0 &  0 &  0\\
   -1 &  1 &  0 &  0\\
    0 & 1 & -1 &  1 \\
    0 & 1 &  0 & 1
\end{pmatrix}$
};
\end{tikzpicture}

\caption{
We represent circular coordinates as explained in \cref{figure:example-coloring}.
\textit{Left:} Four circle-valued maps obtained by running the (Sparse) Circular Coordinates Algorithm on four generators of the first cohomology of a genus two surface.
The generators were obtained using persistent cohomology.
Although the cohomology classes are linearly independent, they do not give a particularly efficient representation of the $1$-dimensional holes in the data: for instance, the first two maps both vary as one goes around the bottom outer hole.
\textit{Right:} Four circle-valued maps obtained by running the (Sparse) Toroidal Coordinates Algorithm, with input the same four cohomology classes used on the left.
\textit{Middle:} The change of basis matrix applied to the cohomology classes in order to geometrically decorrelate them.
See \cref{subsec: genus2surface} for details about this example.}
\label{fig:gen2}

\vspace{0.8cm}

\includegraphics[width=.7\linewidth]{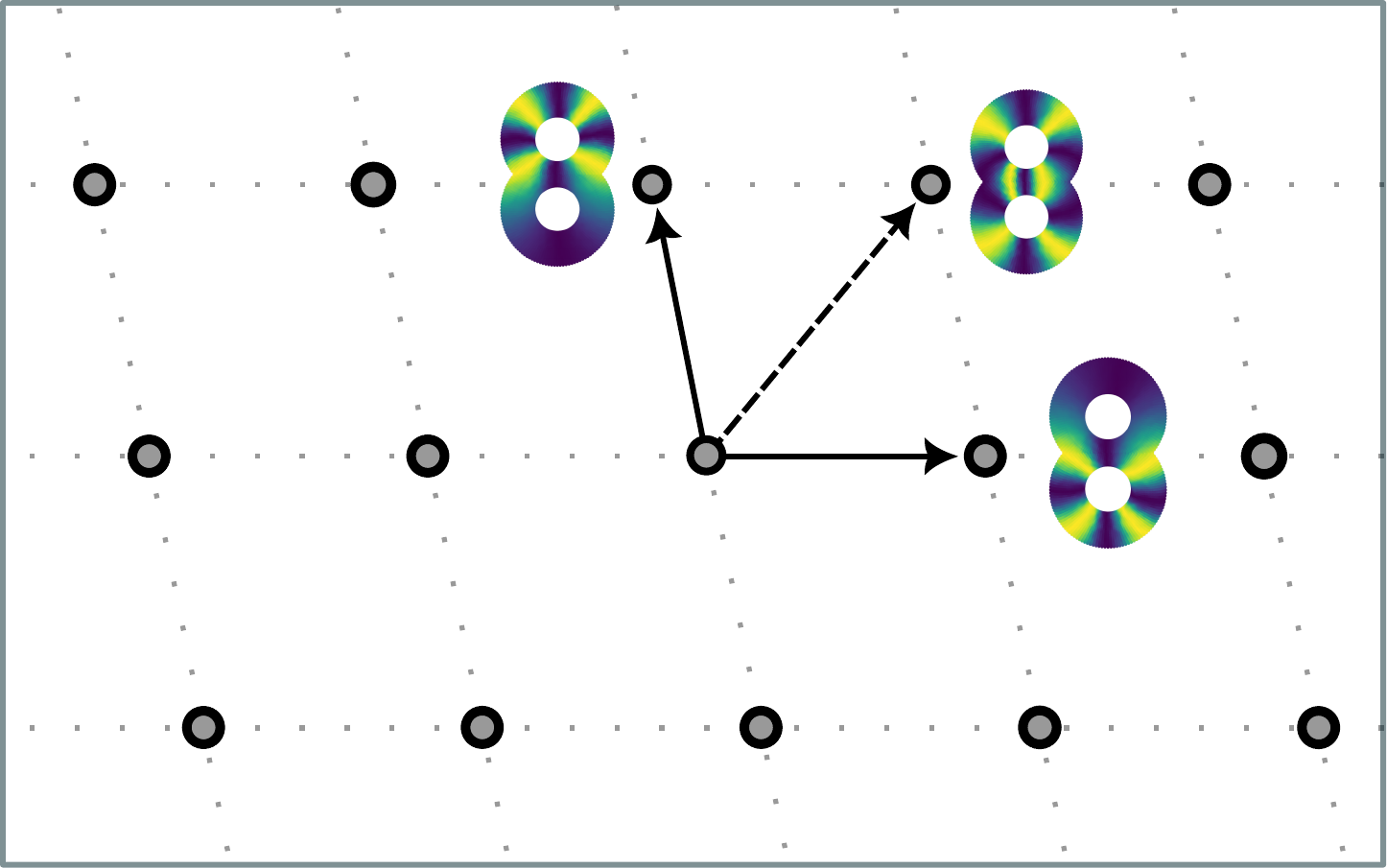}
\caption{The lattice generated by the two left-most circle-valued maps in \cref{fig:gen2}, using our notion of discrete geometric correlation (\cref{definition:geometric-correlation}).
These two circle-valued maps are represented as the horizontal vector and the dashed vector.
Note that, although these two vectors form a basis of the lattice, there exists a basis of smaller total squared length: the one formed by the two solid vectors.
The two solid vectors correspond to two circle-valued returned by the (Sparse) Toroidal Coordinates Algorithm, as shown in \cref{fig:gen2}.}
\label{figure:lattice}
\end{figure}

\subparagraph{Contributions}
Given a Riemannian manifold $\Mcal$, we propose to measure the geometric correlation between smooth maps $f,g: \Mcal \to \circle$ using the Dirichlet form $D(f,g) \in \Rbb$.
We show that, given smooth maps $f,g: \Mcal \to \circle$ obtained by integrating cocycles $\theta$ and $\eta$ defined on the nerve $N(\mathcal{U})$ of an open cover $\mathcal{U}$ of $\Mcal$, there exists an inner product $\innerprod_D$ at the level of cocycles inducing an isometry $\langle \theta , \eta \rangle_D = D(f,g)$ (\cref{proposition:isometry}).
This motivates our Toroidal Coordinates Algorithm (\cref{algorithm:toroidal-coordinates}), which works at the level of cocycles on a simplicial complex and produces low energy torus-valued representations of data.
We prove that the energy minimization subroutine of the Toroidal Coordinates Algorithm is correct (\cref{proposition:algo-eq-1-correct}) and give a geometric interpretation (\cref{proposition:analogy-toroidal-coordinates}).
We introduce the Sparse Toroidal Coordinates Algorithm (\cref{algorithm:sparse-toroidal-coordinates})---a more scalable version of our main algorithm---implemented in \cite{implementation-2}, and showcase it on four datasets (\cref{section:examples}). 

\subparagraph{Structure of the paper}
\cref{section:background} contains background and can be referred to as needed.
The next two sections, \ref{section:toroidal-coordinates-section} and \ref{section:geometric-interpretation}, can be read in any order:
\cref{section:toroidal-coordinates-section} contains a computational description of the Toroidal Coordinates Algorithm, while \cref{section:geometric-interpretation} describes an analogous procedure for Riemannian manifolds and serves as motivation.
\cref{section:sparse-toroidal-coordinates} describes the Sparse Toroidal Coordinates Algorithm, then demonstrated in the examples of \cref{section:examples}.

\subparagraph{Discussion}
In the examples in \cref{section:examples}, running the Sparse Toroidal Coordinates Algorithm on a set of cohomology classes gives results that are qualitatively and quantitatively better than the results obtained by running the Sparse Circular Coordinates Algorithm separately on each class.
This suggests that the Dirichlet form is indeed a useful notion of geometric correlation that can be leveraged for producing geometrically efficient and topologically faithful low-dimensional representations of data.
We believe our methods can be extended to representations valued in non-trivial spaces other than tori, such as other Lie groups.

Various interesting problems remain open:
Is our lattice reduction problem (\cref{problem:our-lattice-reduction}) provably a hard computational problem?
Why is it that the de Silva--Morozov--Vejdemo-Johansson inner product and the inner product estimated in \cref{construction:heuristic} give such similar results (\cref{remark:inner-product-used})?
Are our heuristics for estimating the Dirichlet form from finite samples (Construction~\ref{construction:heuristic} and
\ref{construction:heuristic-dirichlet-form})
consistent?
Here, consistency refers to convergence in probability to the Dirichlet form as the number of samples goes to infinity.

\section{Background}
\label{section:background}

For details about the basics of algebraic topology and Riemannian geometry, we refer the reader to \cite{munkres} and \cite{jost}, respectively.

\subparagraph{Cohomology}
Let $K$ be a finite abstract simplicial complex and let $A$ be either of the rings $\Zbb$ or $\Rbb$.
Let $K_0$ denote the set of vertices of $K$ and let $K_1 = \{(i,j) \in K_0 \times K_0 : \{i,j\} \text{ is a $1$-simplex of $K$}\}$. For a function $\theta : K_1 \to A$, we denote the evaluation of $\theta$ on a pair $(i,j)$ by $\theta^{ij}$.
The group of $0$-cochains $\simpchzero(K;A)$ is the Abelian group of functions $K_0 \to A$, and the group of $1$-cocycles is the Abelian group
\[
    \simpco(K;A) = \left\{ \theta : K_1 \to A\,
    \left|
    \begin{array}{cc}
        \theta^{ij} = -\theta^{ji} \text{ for all $(i,j) \in K_1$},\\
        \theta^{ij} + \theta^{jk} = \theta^{ik}  \text{ for every $2$-simplex $\{i,j,k\}$ of $K$ }
    \end{array}
    \right.
    \right\},
\]
The \define{first cohomology group} of $K$ with coefficients in $A$ is $\simphom(K;A) \coloneqq \simpco(K;A)/\Im(\delta)$, 
where $\delta$ denotes the group morphism $\simpchzero(K;A) \to \simpco(K;A)$ defined by $\delta(\tau)^{ij} = \tau(j) - \tau(i)$.
Given $\theta \in \simpco(K;A)$ we denote its image in $\simphom(K;A)$ as $[\theta] \in \simphom(K;A)$.

For any topological space $B$, we let $\iota$ denote the homomorphism $\iota:\simphom(B;\Zbb) \to \simphom(B;\Rbb)$ induced by the inclusion of coefficients $\Zbb \hookrightarrow \Rbb$.

\subparagraph{The Frobenius inner product}
Let $W$ and $Z$ be real, finite dimensional inner product spaces.
For a linear map $A:Z \to W$, let $A^* : Z \to W$ denotes the adjoint of $A$ with respect to the inner products on $W$ and $Z$.
The \define{Frobenius inner product} between two linear maps $A,B : W \to Z$ is defined as $\langle A, B\rangle_F \coloneqq \trace(A^* B)$.
In particular, the space of linear maps $W \to Z$ can be endowed with the \define{Frobenius norm}, given by $\|A\|_F \coloneqq \sqrt{\trace(A^* A)}$.

\subparagraph{Circle and tori}
We define the circle as the quotient of topological Abelian groups $\circle = \Rbb/\Zbb$, with the induced quotient map $\Rbb \xrightarrow{q} \circle$ given by mapping $r$ to $r\mod \Zbb$.
We endow $\circle$ with the unique Riemannian metric that makes $q$ a local Riemannian isometry.
Given $k \in \Nbb$, let $\torus^k = \left(\circle\right)^k$ denote the \define{$k$-dimensional torus} with the product Riemannian metric.

\subparagraph{Circle-valued maps}
Let $B$ be a topological space.
Given $f,g : B \to \circle$, define $f+g : B \to \circle$ by $(f+g)(p) = f(p) + g(p)$ for all $p \in B$. 
This endows the set of maps $B \to \circle$ with the structure of an Abelian group.
We say $f : B \to \circle$ and $g : B \to \circle$ are \define{rotationally equivalent} if $f - g$ is constant on each connected component of $B$.
Analogously, for a simplicial complex $K$,
we say that maps on vertices $f : K_0 \to \circle$ and $g : K_0 \to \circle$ are \define{rotationally equivalent} if $f - g : K_0 \to \circle$ is constant on each connected component of $K$.

\subparagraph{Differential of circle-valued maps}
There is a canonical isomorphism $\tangent \circle \cong \circle \times \Rbb$ of Riemannian vector bundles over $\circle$.
Here, $\circle \times \Rbb \to \circle$ is the trivial Riemannian vector bundle over $\circle$ and the isomorphism is given by the linear isometries $d( \cdot - q)_q :  \tangent_q \circle \to \tangent_0 \circle \cong \Rbb$,
where $\cdot - q : \circle \to \circle$ denotes subtracting $q$, and the isomorphism $\tangent_0 \circle \cong \Rbb$ is chosen once and for all.
Using the isomorphism $\tangent \circle \cong \circle \times \Rbb$, we can unambiguously treat the differential of a map $f : \Mcal \to \circle$ at a point $p \in \Mcal$ as a linear function $d f_p : \tangent_p \Mcal \to \Rbb$.
In particular, any smooth map $f : \Mcal \to \circle$ induces a $1$-form $df \in \Omega^1(\Mcal)$ on $\Mcal$.

\subparagraph{Dirichlet energy and Dirichlet form}
Given a closed Riemannian manifold $\Mcal$, we let $\mu$ denote its Riemannian measure.
The \define{Dirichlet energy} of a smooth map $f : \Mcal \to \Ncal$ between Riemannian manifolds is
\[
    E[f] \coloneqq \frac{1}{2}\int_{p \in \Mcal} \|df_p\|^2_F \;\dsf \mu(p),
\]
where $df_p : \tangent_p\Mcal \to \tangent_{f(p)} \Ncal$ is the differential of $f$, a map between inner product spaces.

Recall that the inner product on the space of $1$-forms $\Omega^1(\Mcal)$ is given, for $\theta, \eta \in \Omega^1(\Mcal)$, by $\langle \theta, \eta \rangle_{\Omega^1} \coloneqq \int_{p \in \Mcal} \langle \theta_p, \eta_p \rangle_F\; \dsf \mu(p)$.
One can thus extend the Dirichlet energy of circle-valued maps to a bilinear form, as follows.
Given $f,g : \Mcal \to \circle$, define their \define{Dirichlet form} as
\[
    D(f,g) \coloneqq \frac{1}{2} \langle df, dg \rangle_{\Omega^1} = \frac{1}{2} \int_{p \in \Mcal} \langle df_p, dg_p \rangle_F \;\dsf \mu(p).
\]
We remark that, as defined, the Dirichlet form makes sense only for circle-valued maps.
We conclude by noticing that the Dirichlet form and the Dirichlet energy determine each other.
On one hand, we have $E[f] = D(f,f)$.
On the other hand, we have $D(f,g) = \frac{1}{4}(E(f + g) - E(f - g))$, by the polarization identity for the inner product space $\Omega^1(\Mcal)$.

\section{The Toroidal Coordinates Algorithm}
\label{section:toroidal-coordinates-section}

\subsection{From circular coordinates to toroidal coordinates}
\label{section:main-algorithms}
We recall the circular coordinates algorithm of \cite{silva-vejdemo,silva-morozov-vejdemo} and use its main minimization subroutine to motivate the Toroidal Coordinates Algorithm.
The most relevant portion of the full pipeline%
\footnote{We refer the reader to \cite[Sections~2.2--2.4]{silva-morozov-vejdemo} for details about the rest of the pipeline.} is given as \cref{algorithm:circular-coordinates}, which we refer to as the \define{Circular Coordinates Algorithm}.

As can be easily checked, the minimization subroutine (\cref{algorithm:harmonic-representative}) of the Circular Coordinates Algorithm returns a solution to the following problem:

\begin{problem}
    \label{problem:minimization-circular-coordinates}
Given $0\neq \alpha \in \simphom(K;\Zbb)$ and an inner product $\innerprod$ on $\simpco(K;\Rbb)$, find $\theta \in \simpco(K;\Rbb)$ of minimum norm such that $[\theta] = \iota(\alpha) \in \simphom(K;\Rbb)$.
\end{problem}

We propose the following extension of \cref{problem:minimization-circular-coordinates}
to the case in which more than one cohomology class is selected.

\begin{problem}
    \label{equation:minimization-toroidal-coordinates}
Given linearly independent $\alpha_1, \dots, \alpha_k \in \simphom(K;\Zbb)$ and inner product $\innerprod$ on $\simpco(K;\Rbb)$, find $\theta_1 , \dots, \theta_k \in \simpco(K;\Rbb)$ minimizing $\sum_{j=1}^k \|\theta_j\|^2$, with the property that
the sets $\{[\theta_j]\}_{1 \leq j \leq k}$ and $\{\iota(\alpha_j)\}_{1 \leq j \leq k}$ generate the same Abelian subgroup of $\simphom(K;\Rbb)$.
\end{problem}

Simple examples, such as the one depicted in \cref{figure:lattice}, show that
\cref{equation:minimization-toroidal-coordinates} does not reduce to solving \cref{problem:minimization-circular-coordinates} for each individual cohomology class.
Indeed, as explained in \cref{section:lattice-reduction}, we believe that \cref{equation:minimization-toroidal-coordinates} is significantly harder to solve exactly than \cref{problem:minimization-circular-coordinates}.
Nevertheless, we also show that one can use the Lenstra--Lenstra--Lov\'{a}sz lattice basis reduction algorithm to find an approximate solution to \cref{equation:minimization-toroidal-coordinates}.
This approximation is the content of the following result, which is proven in 
\cref{section:proofs-lattice-reduction}. 

\begin{theorem}
    \label{proposition:algo-eq-1-correct}
    The output of \cref{algorithm:solving-eq-1} consists of cocycles $\theta_1, \dots, \theta_k$ such that $\sum_{j=1}^k \|\theta_j\|^2$ is at most $2^{k-1}$ times the optimal solution of \cref{equation:minimization-toroidal-coordinates}.
\end{theorem}

\cref{algorithm:solving-eq-1} constitutes the main minimization subroutine of the Toroidal Coordinates Algorithm, which is given as \cref{algorithm:toroidal-coordinates}.

\subsection{On the choice of inner product}
\cref{algorithm:circular-coordinates,algorithm:toroidal-coordinates} depend on a user-given choice of inner product on $\simpco(K;\Rbb)$.
In \cite{silva-vejdemo,silva-morozov-vejdemo}, the inner product used is given by
\begin{equation}
    \label{equation:SMV-inner-product}
     \langle \theta, \eta \rangle_{\SMV} \coloneqq \sum_{\{i,j\} \in K_1} \theta^{ij} \eta^{ij}.
\end{equation}
The motivation for this choice is given in \cite[Proposition~2]{silva-morozov-vejdemo}, which implies that the map $K_0 \to \circle$ returned by \cref{algorithm:circular-coordinates} has the property that it can be extended to a continuous function $|K| \to \circle$ which maps each edge $\{i,j\}$ of $K$ to a curve of length $|\theta^{ij}|$.
Thus, with this choice of inner product, the circle-valued representation returned by \cref{algorithm:circular-coordinates} is one that stretches the edges of the simplicial complex as little as possible.

There are other natural choices of inner product.
In particular, we show in \cref{proposition:isometry} that there exists an inner product between cocycles that recovers the Dirichlet form between circle-valued maps obtained by integrating these cocycles.
Since, as explained in the contributions section, we propose to measure the geometric correlation between maps $f,g : \Mcal \to \circle$ on a Riemannian manifold using their Dirichlet form, this motivates the following definition.

\begin{definition}
\label{definition:geometric-correlation}
Let $K$ be a simplicial complex and let $\innerprod$ be an inner product on $\simpco(K;\Rbb)$.
Given cocycles $\theta, \eta \in \simpco(K;\Rbb)$ with $[\theta], [\eta] \in \Im(\iota : \simphom(K;\Zbb) \to \simphom(K;\Rbb))$, define the \define{discrete geometric correlation} between $\integrate_\theta, \integrate_\eta : K_0 \to \circle$ as $\langle \theta, \eta \rangle$.
Here $\integrate$ is as defined in \cref{algorithm:cocycle-integration}.
\end{definition}

In \cref{section:geometric-interpretation}, we give a geometric interpretation of the Toroidal Coordinates Algorithm and provide more details as to why the above notion of discrete geometric correlation is a discrete analogue of the Dirichlet form (\cref{remark:inner-product-is-dirichlet-energy}).

We conclude this section with a remark explaining why an exact or approximate solution to \cref{equation:minimization-toroidal-coordinates} promotes low discrete geometric correlation.

\begin{remark}
    \label{remark:minimimum-has-low-correlation}
Let $\theta_1, \dots, \theta_k \in \simpco(K;\Rbb)$.
For any linear map $A : W \to Z$ between finite dimensional inner product spaces, we have $\|A^*A\|_F \leq \|A\|_F^2$.
Thus, if $A : \Rbb^k \to \simpco(K;\Rbb)$ is given by mapping the $j$th standard basis vector to $\theta_j$, we get
\[
    \sum_{1 \leq i,j \leq k} \langle \theta_i, \theta_j \rangle^2 = \|A^*A\|_F \leq \|A\|_F^2 = \sum_{j=1}^k \|\theta_j\|^2.
\]
This implies that a set of cocycles solving \cref{equation:minimization-toroidal-coordinates}
exactly or approximately (right-hand side) induces, by integration (\cref{algorithm:cocycle-integration}), a set of cicle-valued maps with low pairwise squared discrete geometric correlation (left-hand side).
\end{remark}

\begin{figure}
\begin{algorithm}[H]
\setstretch{1.35}
    \begin{algorithmic}[1]
    \Require{a non-trivial cohomology class $\alpha \in \simphom(K;\Zbb)$ and an inner product $\innerprod$ on $\simpco(K;\Rbb)$}
    \Ensure{a function $\circcoords_{\alpha} : K_0 \to \circle$}
    \State Let $\theta \coloneqq \harmrep(\alpha, \innerprod)$
    \State Let $\circcoords_{\alpha} \coloneqq \integrate_\theta$
    \end{algorithmic}
        \caption{The Circular Coordinates Algorithm}
        \label{algorithm:circular-coordinates}
\end{algorithm}

\begin{algorithm}[H]
\setstretch{1.35}
    \begin{algorithmic}[1]
    \Require{l.i.~cohomology classes $\alpha_1, \dots, \alpha_k \in \simphom(K;\Zbb)$ and inner product $\innerprod$ on $\simpco(K;\Rbb)$}
    \Ensure{a function $\torcoords_{\alpha} : K_0 \to \torus^k$}
    \State Let $\theta_1, \dots, \theta_k \coloneqq \lowenergyreps(\alpha_1, \dots, \alpha_k, \innerprod)$
    \State Let $\torcoords_{\alpha}  \coloneqq (\integrate_{\theta_1}, \dots, \integrate_{\theta_k})$
    \end{algorithmic}
        \caption{The Toroidal Coordinates Algorithm}
        \label{algorithm:toroidal-coordinates}
\end{algorithm}

\begin{algorithm}[H]
\setstretch{1.35}
    \begin{algorithmic}[1]
    \Require{a non-trivial cohomology class $\alpha \in \simphom(K;\Zbb)$ and an inner product $\innerprod$ on $\simpco(K;\Rbb)$}
    \Ensure{a cocycle $\harmrep(\alpha,\innerprod) \in \simpco(K;\Rbb)$}
    \State Let $\eta \in \simpco(K;\Zbb)$ be such that $[\eta] = \alpha \in \simphom(K;\Zbb)$
    \State Use least squares, w.r.t.~$\innerprod$, to solve $\tau = \argmin \{\;\|\iota(\eta) - \delta(\tau)\| \;\mid\; \tau : K_0 \to \Rbb\;\}$
    \State Let $\harmrep(\alpha,\innerprod) \coloneqq \iota(\eta) - \delta(\tau)$
    \end{algorithmic}
        \caption{Harmonic representative}
        \label{algorithm:harmonic-representative}
\end{algorithm}

\begin{algorithm}[H]
\setstretch{1.35}
    \begin{algorithmic}[1]
    \Require{l.i.~cohomology classes $\alpha_1, \dots, \alpha_k \in \simphom(K;\Zbb)$ and inner product $\innerprod$ on $\simpco(K;\Rbb)$}
    \Ensure{list of $k$ cocycles $\lowenergyreps(\alpha_1, \dots, \alpha_k, \innerprod) \subseteq \simpco(K;\Rbb)$}
    \State Let $\eta_j \coloneqq \harmrep(\alpha_j,\innerprod)$ for $1 \leq j \leq k$
    \State Compute the Cholesky decomposition $G = C C^*$ of $G \in \Rbb^{k \times k}$ with $G_{ij} = \langle \eta_i, \eta_j \rangle$
    \State Let $b_1, \dots, b_k \coloneqq \lll(C_1, \dots, C_k)$, with $C_j$ the $j$th row of $C$ and $\lll$ as in \cref{section:lattice-reduction}
    \State Let $M \in \Zbb^{k\times k}$ be the change of basis matrix such that $M C = (b_1, \dots, b_k)^T$
    \State Let $\lowenergyreps(\alpha_1, \dots, \alpha_k, \innerprod) \coloneqq M \; (\eta_1, \dots, \eta_k)^T$
    \end{algorithmic}
        \caption{Low energy representatives}
        \label{algorithm:solving-eq-1}
\end{algorithm}

\begin{algorithm}[H]
\setstretch{1.35}
    \begin{algorithmic}[1]
    \Require{a cocycle $\theta \in \simpco(K;\Rbb)$ such that $[\theta] \in \Im(\iota : \simphom(K;\Zbb) \to \simphom(K;\Rbb))$}
    \Ensure{a function $\integrate_\theta : K_0 \to \circle$}
    \State Assume $K$ is connected, otherwise do the following in each connected component
    \State Choose $x \in K_0$ arbitrarily
    \State \textbf{for} $y \in K_0$ \textbf{do}
    \State \hspace{0.5cm} Choose a path $x = y_0, y_1, \dots, y_{\ell - 1}, y_\ell = y$ from $x$ to $y$, arbitrarily
    \State \hspace{0.5cm} Let $\integrate_\theta(y) \coloneqq \left(\theta^{y_0 y_{1}} + \theta^{y_1 y_{2}} + \dots + \theta^{y_{\ell-2} y_{\ell-1}} + \theta^{y_{\ell-1} y_{\ell}}\right) \mod \Zbb$
    \end{algorithmic}
        \caption{Cocycle integration}
        \label{algorithm:cocycle-integration}
\end{algorithm}

\end{figure}

\subsection{Minimizing the objective function with lattice reduction}
\label{section:lattice-reduction}
We start by describing the specific lattice reduction problem we are interested in.
Fix $k \in \Nbb$ and a $k$-dimensional real vector space $R$ with an inner product.
A full-dimensional \define{lattice} $L$ in $R$ is a discrete subgroup $L \subseteq R$ which generates $R$ as a real vector space.
An ordered \define{basis} of a lattice $L \subseteq R$ consists of an ordered list $B = \{b_1, \dots, b_k\} \subseteq L$ of linearly independent vectors that generate $L$ as an Abelian group.
We are interested in the following problem.

\begin{problem}
\label{problem:our-lattice-reduction}
Let $L \subseteq R$ be a lattice.
Find a basis $B$ of $L$ minimizing $\|B\|_F^2 = \sum_{i=1}^k \|b_i\|^2$.
\end{problem}

We suspect that \cref{problem:our-lattice-reduction} is in general hard to solve exactly or approximately up to a small multiplicative constant.
Formally establishing that this problem is hard is beyond the scope of this work since hardness results for these kinds of problems---like \cite{ajtai} for the shortest vector problem---are usually quite involved; we refer the reader to \cite{khot,regev} for surveys.
We note that minimizations like the one in \cref{problem:our-lattice-reduction} have already been considered in the computational number theory literature, see, e.g., \cite[Equation~38]{dayal-varanasi}.


We content ourselves with the following result, which shows the Lenstra--Lenstra--Lov\'{a}sz lattice basis reduction algorithm (LLL-algorithm), a polynomial-time algorithm introduced in \cite{lenstra-lenstra-lovasz}, provides an approximate solution to \cref{problem:our-lattice-reduction}.
For our purposes, the LLL-algorithm takes as input linearly independent vectors $\{b_1, \dots, b_k\}$ in $\Rbb^k$ and returns a \define{reduced} basis, which we denote by $\lll(b_1, \dots, b_k)$.
We shall not recall the definition of reduced basis here, since all we need to know about them is the following.

\begin{lemma}
    \label{proposition:LLL-gives-approximate-solution}
    Let $L \subseteq \Rbb^n$ and let $V$ be a solution to \cref{problem:our-lattice-reduction} for $L \subseteq \Rbb^n$.
    If $B$ is an reduced basis, then $\|B\|_F^2 \leq 2^{k-1}\,\|V\|_F^2$.
\end{lemma}

We prove \cref{proposition:LLL-gives-approximate-solution} in
\cref{section:proofs-lattice-reduction},
where we use it to prove \cref{proposition:algo-eq-1-correct}.
We conclude this section with a practical remark about the LLL-algorithm.

\begin{remark}
    Although the LLL-algorithm can be run with  any input $\{b_1, \dots, b_n\} \subseteq L \subseteq \Rbb^n$, it is guaranteed to terminate only if one uses infinite precision arithmetic.
    In \cite{lenstra-lenstra-lovasz}, this is dealt with by assuming that the given lattice has rational coordinates, i.e., $L \subseteq \Qbb^n \subseteq \Rbb^n$; see \cite[Remark~1.38]{lenstra-lenstra-lovasz}.
    This is a reasonable assumption in our case, since we expect to be given the input cocycles and inner product with some finite precision.

    In our implementation of the LLL-algorithm, we use floating-point arithmetic, for simplicity, 
    and this did not present any problems to us.
    We note that floating-point algorithms with polynomial guarantees do exist in the case $L \subseteq \Zbb^n \subseteq \Rbb^n$, see, e.g., \cite{nguen-stehle}.
\end{remark}

\section{Geometric Interpretation of the Toroidal Coordinates Algorithm}
\label{section:geometric-interpretation}

Let $\Mcal$ be a closed Riemannian manifold.
We propose the following problem as a suitable objective for finding an efficient representation of $\Mcal$ which captures any chosen set of $1$-dimensional holes of $\Mcal$.

\begin{problem}
    \label{problem:geometric-objective-toroidal-coords}
    Given linearly independent cohomology classes $\alpha_1, \dots, \alpha_k \in \simphom(\Mcal;\Zbb)$, find a smooth map $f : \Mcal \to \torus^k$ of minimum Dirichlet energy, with the property that the induced morphism $f^* : \simphom(\torus^k;\Zbb) \to \simphom(\Mcal;\Zbb)$ restricts to an isomorphism between $\simphom(\torus^k;\Zbb) \cong \Zbb^k$ and the subgroup of $\simphom(\Mcal;\Zbb)$ generated by $\alpha_1, \dots, \alpha_k$.
\end{problem}

In this section, we show that the above problem can be solved by an analogue of our Toroidal Coordinates Algorithm, thus providing a geometric interpretation of our algorithm.
In \cref{remark:inner-product-is-dirichlet-energy}, at the end of this section, we explain how this interpretation motivates the notion of discrete geometric correlation of \cref{definition:geometric-correlation}.

First, we give the analogue of cocycle integration (\cref{algorithm:cocycle-integration}) for $1$-forms.

\begin{construction}
    \label{construction:continuous-integration}
    Given a closed $1$-form $\theta \in \Omega^1(\Mcal)$ such that $[\theta] \in \Im(\simphom(\Mcal;\Zbb) \to \simphom(\Mcal;\Rbb))$, consider the following procedure, which returns a function $f : \Mcal \to \circle$.
    \begin{enumerate}
    \item Assume $\Mcal$ is connected, otherwise do the following in each connected component.
    \item Choose $x \in \Mcal$ arbitrarily.
    \item For each $y \in \Mcal$, let $p : [0,1] \to \Mcal$ be any smooth path from $x$ to $y$.
    \item For each $y \in \Mcal$, define $f(y) = \left(\int_0^1 \theta_{p(t)}(p'(t)) dt\right) \mod \mathbb{Z}$.
    \end{enumerate}
\end{construction}


It is worth remarking that, although \cref{construction:continuous-integration} depends on arbitrary choices, all choices yield rotationally equivalent outputs.

The following procedure is the analogue of the Toroidal Coordinates Algorithm.

\begin{construction}
    \label{construction:toroidal-coords-riemannian}
Given linearly independent cohomology classes $\alpha_1, \dots, \alpha_k \in \simphom(\Mcal;\Zbb)$, consider the following procedure, which returns a function $f : \Mcal \to \torus^k$.
\begin{enumerate}
    \item Find closed $\theta_1 , \dots, \theta_k \in \Omega^1(\Mcal)$ minimizing $\sum_{j=1}^k \|\theta_j\|^2$, with the property that the sets $\{[\theta_j]\}_{1 \leq j \leq k}$ and $\{\iota(\alpha_j)\}_{1 \leq j \leq k}$ generate the same Abelian subgroup of $\simphom(\Mcal;\Rbb)$.
    \item Return $(f_1, \dots, f_k) : \Mcal \to \torus^k$, where $f_j$ is obtained by integrating $\theta_j$ (\cref{construction:continuous-integration}).
\end{enumerate}
\end{construction}

\begin{proposition}
    \label{proposition:analogy-toroidal-coordinates}
    \cref{construction:toroidal-coords-riemannian}
    returns a solution to \cref{problem:geometric-objective-toroidal-coords}.
\end{proposition}

A proof of \cref{proposition:analogy-toroidal-coordinates} is in
\cref{section:proof-of-analogy-toroidal-coords}.
We conclude with a remark relating the Dirichlet form to our notion of discrete geometric correlation.

\begin{remark}
    \label{remark:inner-product-is-dirichlet-energy}
    Recall from the contributions section that we propose to measure geometric correlation between maps $f,g : \Mcal \to \circle$ using the Dirichlet form $D(f,g)$.
    On one hand, if $f$ and $g$ are obtained using \cref{construction:continuous-integration} with input $1$-forms $\theta$ and $\eta$, respectively, then $D(f,g) = \frac{1}{2}\langle \theta, \eta \rangle_{\Omega^1}$, by
    \cref{lemma:differential-of-circlular-coordinates}.
    On the other hand, given a simplicial complex $K$ with inner product $\innerprod$ on $\simpco(K;\Rbb)$, and maps $f', g' : K_0 \to \circle$ obtained using \cref{algorithm:cocycle-integration} with inputs cocycles $\theta'$ and $\eta'$, respectively, we defined the discrete geometric correlation between $f'$ and $g'$ as $\langle \theta', \eta' \rangle$.
    In this sense, our notion of discrete geometric correlation is a discrete analogue of the Dirichlet form. \cref{proposition:isometry} makes this analogy precise: when using Algorithm ~\ref{algorithm:sparse-cocycle-integration}, there exists an inner product that exactly recovers the Dirichlet form.
\end{remark}

\section{The Sparse Toroidal Coordinates Algorithm}
\label{section:sparse-toroidal-coordinates}

Although effective, the Circular Coordinates Algorithm has two practical drawbacks.
First, the simplicial complex $K$ is usually taken to be a Vietoris--Rips complex, and thus the cohomology computations scale with the number of data points.
Second, the circle-valued representation returned by the algorithm is defined only on the input data and no representation is provided for out-of-sample data points.
The sparse circular coordinates algorithm of \cite{perea2020sparse} addresses these shortcomings.
We now describe a version of the sparse circular coordinates algorithm\footnote{We refer the reader to \cite{perea2020sparse} for the full pipeline.}
and recall the steps not included here when describing \cref{construction:pipeline-examples} in the examples.

\begin{algorithm}[H]
\setstretch{1.35}
    \begin{algorithmic}[1]
    \Require{a finite open cover $\Ucal = \{U_x\}_{x \in I}$ of a topological space $B$, a partition of unity $\Phi =\{\phi_x\}_{x \in I}$ subordinate to $\Ucal$,
    a simplicial complex $K \supseteq N(\Ucal)$, and a cocycle $\theta \in \simpco(K;\Rbb)$ such that $[\theta] \in \Im(\iota : \simphom(K;\Zbb) \to \simphom(K;\Rbb))$}
    \Ensure{a function $\sparseintegrate^\Phi_\theta : B \to \circle$}
    \State Assume $K$ is connected, otherwise do the following in each connected component
    \State Choose $x \in K_0$ arbitrarily
    \State \textbf{for} $y \in K_0$ \textbf{do}
    \State \hspace{0.5cm} Choose a path $x = y_0, y_1, \dots, y_{\ell - 1}, y_\ell = y$ from $x$ to $y$, arbitrarily
    \State \hspace{0.5cm} Let $\tau_y \coloneqq \theta^{y_0 y_{1}} + \theta^{y_1 y_{2}} + \dots + \theta^{y_{\ell-2} y_{\ell-1}} + \theta^{y_{\ell-1} y_{\ell}}$
    \State Let $\sparseintegrate^\Phi_\theta(b) \coloneqq \left(\tau_y + \sum_{z \in I} \phi_z(b) \; \theta^{yz}\right) \mod \Zbb$, where $b \in U_y$
    \end{algorithmic}
        \caption{Sparse cocycle integration}
        \label{algorithm:sparse-cocycle-integration}
\end{algorithm}


\begin{algorithm}[H]
\setstretch{1.35}
    \begin{algorithmic}[1]
    \Require{
    a finite open cover $\Ucal = \{U_x\}_{x \in I}$ of a topological space $B$, a partition of unity $\Phi = \{\phi_x\}_{x \in I}$ subordinate to $\Ucal$,
    a simplicial complex $K \supseteq N(\Ucal)$,    
    a non-trivial cohomology class $\alpha \in \simphom(K;\Zbb)$, and an inner product $\innerprod$ on $\simpco(K;\Rbb)$}
    \Ensure{a function $\sparsecirccoords_{\alpha} : B \to \circle$}
    \State Let $\theta \coloneqq \harmrep(\alpha, \innerprod)$
    \State Let $\circcoords_{\alpha} \coloneqq \sparseintegrate^\Phi_\theta$
    \end{algorithmic}
        \caption{The Sparse Circular Coordinates Algorithm}
        \label{algorithm:sparse-circular-coordinates}
\end{algorithm}

\begin{algorithm}[H]
\setstretch{1.35}
    \begin{algorithmic}[1]
    \Require{
    a finite open cover $\Ucal = \{U_x\}_{x \in I}$ of a topological space $B$, a partition of unity $\Phi = \{\phi_x\}_{x \in I}$ subordinate to $\Ucal$,
    a simplicial complex $K \supseteq N(\Ucal)$,    
    l.i.~cohomology classes $\alpha_1, \dots, \alpha_k \in \simphom(K;\Zbb)$, and inner product $\innerprod$ on $\simpco(K;\Rbb)$}
    \Ensure{a function $\sparsetorcoords_{\alpha} : B \to \torus^k$}
    \State Let $\theta_1, \dots, \theta_k \coloneqq \lowenergyreps(\alpha_1, \dots, \alpha_k, \innerprod)$
    \State Let $\sparsetorcoords_{\alpha}  \coloneqq (\sparseintegrate^\Phi_{\theta_1}, \dots, \sparseintegrate^\Phi_{\theta_k})$
    \end{algorithmic}
        \caption{The Sparse Toroidal Coordinates Algorithm}
        \label{algorithm:sparse-toroidal-coordinates}
\end{algorithm}

As in previous cases, we remark that, although the sparse cocycle integration subroutine (\cref{algorithm:sparse-cocycle-integration}) depends on arbitrary choices, all choices yield rotationally equivalent outputs.

\medskip

We now show that, when $B$ is a closed Riemannian manifold, there is a choice of inner product on cocycles that coincides with the Dirichlet form between the corresponding circle-valued maps, making the analogy in \cref{remark:inner-product-is-dirichlet-energy} formal.

\begin{definition}
    Let $\Ucal = \{U_x\}_{x \in I}$ be a finite open cover of a closed Riemannian manifold $\Mcal$ and let $\Phi = \{\phi_x\}_{x \in I}$ be a smooth partition of unity subordinate to $\Ucal$.
    Define the inner product $\innerprod_D$ on $\simpco(N(\Ucal);\Rbb)$ by 
    \[
        \langle\theta,\eta\rangle_D \coloneqq
        \frac{1}{2} \, \sum_{w,y,z \in I} D_{wyz}\; \theta^{wy}\, \eta^{wz},\;\;
        \text{where}\;\; D_{wyz} \coloneqq \int_{b \in \Mcal} \;\langle d(\phi_y)_b, d(\phi_z)_b \rangle_F\; \phi_w(b)\; \dsf\mu(b).
    \]
\end{definition}
Note that the quantities $D_{wyz}$ do not depend on the cocycles.

\begin{theorem}
    \label{proposition:isometry}
    Let $\Ucal = \{U_x\}_{x \in I}$ be a finite open cover of a closed Riemannian manifold $\Mcal$, let $K \supseteq N(\Ucal)$, and let $\Phi = \{\phi_x\}_{x \in I}$ be a smooth partition of unity subordinate to $\Ucal$.
    Assume $\theta,\eta \in \simpco(K;\Rbb)$ are such that $[\theta], [\eta] \in \simphom(K;\Rbb)$ are in the image of $\iota : \simphom(K;\Zbb) \to \simphom(K;\Rbb)$.
    Let $f = \sparseintegrate_\theta^\Phi : \Mcal \to \circle$ and 
    $g =  \sparseintegrate_\eta^\Phi$.
    Then, $f$ and $g$ are smooth and $D(f,g) = \langle \theta, \eta\rangle_D$.
\end{theorem}
We prove \cref{proposition:isometry} in
\cref{section:proof-of-isometry}.
We conclude by giving a heuristic for computing an estimate $\innerprod_{\widehat{D}}$ of $\innerprod_D$.
Addressing the consistency of this heuristic is left for future work.

\begin{construction}
    \label{construction:heuristic}
    Let $X \subseteq \Mcal \subseteq \Rbb^n$ be a finite sample of a smoothly embedded closed manifold.
    Assume given a subsample $I \subseteq X$ as well as $\epsilon > 0$ such that $\Mcal \subseteq \bigcup_{x \in I} B(x,\epsilon)$.
    For $w,y,z \in I$, we seek to estimate $D_{wyz}$, where the open cover is taken to be $\Ucal = \{B(x,\epsilon)\cap \Mcal\}_{x \in I}$ and $\Phi = \{\phi_x\}_{x \in I}$ is a smooth partition of unity subordinate to $\Ucal$.
    \begin{enumerate}
        \item Form a neighborhood graph $G$ on $X_w \coloneqq X \cap B(w,\epsilon)$.
        For instance, this can be done by selecting $k \in \Nbb$ and using an undirected $k$-nearest neighbor graph.
        \item Compute weights $h(a,b) \geq 0$ for the edges $(a,b) \in G$.
        For instance, this can be done by selecting a radius $\delta > 0$ and letting $h(a,b) = \exp(-\|a-b\|^2/\delta^2)$.
        \item For $a \in X_w$, let $N(a) = \{b \in G \mid (a,b) \in G\}$, and define 
        \[
            \widehat{D_{wyz}} = \sum_{a \in G}  \left( \frac{1}{N(a)}\sum_{b \in N(a)} h(a,b) \; (\phi_y(b) - \phi_y(a)) \; (\phi_z(b) - \phi_z(a)) \right) \phi_w(a).
        \]
    \end{enumerate}
\end{construction}

\begin{remark}
    \label{remark:inner-product-used}
We have implemented the estimated inner product $\innerprod_{\widehat{D}}$ in \cite{implementation-2}.
In all examples we have considered, running the algorithms in this paper with inner product $\innerprod_{\widehat{D}}$ on one hand, and with the 
de Silva, Morozov, and Vejdemo-Johansson inner product $\innerprod_{\SMV}$ (\cref{equation:SMV-inner-product}) on the other, gives results that are essentially indistinguishable.
For this reason, and for concreteness, in \cref{section:examples} we use $\innerprod_{\SMV}$.
We leave the question of when and why the two inner products give such similar results for future work.
\end{remark}

\section{Examples}
\label{section:examples}

We compare the output of the Sparse Circular Coordinates Algorithm~\cite{perea2020sparse} run independently on several cohomology classes with that of the Sparse Toroidal Coordinates Algorithm.
We use the DREiMac~\cite{tralie2021dreimac} implementation of the Sparse Circular Coordinates Algorithm and our extension implementing the Sparse Toroidal Coordinates Algorithm.
The code together with Jupyter notebooks replicating the examples here can be found at \cite{implementation-2}.

The examples include a synthetic genus two surface (Sec.~\ref{subsec: genus2surface}), a dataset from \cite{LEDERMAN2018509} of two figurines rotating at different speeds (Sec.~\ref{section:dog-yoda}), a solution set of the Kuramoto--Sivashinsky equation obtained with Mathematica \cite{mathematica} (Sec.~\ref{section:KS}), and a synthetic dataset modeling neurons tuned to head movement of bats (Sec.~\ref{section:neuroscience}). 
We use the following pipeline.

\begin{pipeline}
    \label{construction:pipeline-examples}
Assume we are given a point cloud $X \subseteq \Rbb^n$.
\begin{enumerate}
    \item Compute a subsample $I \subseteq X$ using maxmin sampling (see \cite{DYER1985, GONZALEZ1985, perea2020sparse}).
    \item Fix a large prime $p$; we take $p = 41$.
    \item Compute Vietoris--Rips persistent cohomology of $I$ in degree $1$ with coefficients in $\Zbb/p\Zbb$.
    \item Looking at the persistence diagram, identify a filtration step $\epsilon > 0$ at which cohomology classes $\beta_1, \dots, \beta_k \in \simphom(\rips_\epsilon(I),\Zbb/p\Zbb)$ of interest to the user are alive.
    Do this in such a way that $X \subseteq \cup_{x \in I} B(x,\epsilon/2)$.
    \item Note that $\Ucal = \{B(x,\epsilon/2)\}_{x \in I}$ covers $X$ and define $K \coloneqq \rips_\epsilon(I) \supseteq N(\Ucal)$.
    \item Lift $\beta_1, \dots, \beta_k \in \simphom(K,\Zbb/p\Zbb)$ to classes $\alpha_1, \dots, \alpha_k \in \simphom(K,\Zbb)$ (see \cite[Section~2.4]{silva-morozov-vejdemo}).
    \item Choose a partition of unity subordinate to $\Ucal$ (see \cite[Section~4]{perea2020sparse}).
    We use the inner product $\innerprod_{\SMV}$ on cocycles (as explained in \cref{remark:inner-product-used}).
    \item On one hand, run the Sparse Circular Coordinates Algorithm (\cref{algorithm:sparse-circular-coordinates}) on each class $\alpha_j$ separately, and get $k$ circle-valued maps $X \to \circle$.
    \item On the other hand, run the Sparse Toroidal Coordinates Algorithm (\cref{algorithm:sparse-toroidal-coordinates}) on all classes $\alpha_1, \dots, \alpha_k$ simultaneously, to again get $k$ circle-valued maps $X \to \circle$.
\end{enumerate}
\end{pipeline}

In order to show that the Sparse Toroidal Coordinates Algorithm returns coordinates with lower correlation and energy, we quantify the performance of the two algorithms using the estimated Dirichlet correlation matrix
(see \cref{section:estimatingdirichlet})
of the circle-valued functions obtained from them.
When the functions are obtained from the Sparse Circular Coordinates Algorithm (resp.~Sparse Toroidal Coordinates Algorithm), we denote the correlation matrix by $D_{SCC}$ (resp.~$D_{STC}$). Note that diagonal correlation matrices reflect complete independence of coordinates. Hence, we interpret correlations matrices that are close to being diagonal as indicating low correlation and high independence of recovered coordinates.

The correlation computations depend on two parameters (a $k$ for a $k$-nearest neighbor graph and a choice of edge weights). We use $k = 15$ and weights related to the scale of the data, but note that the results are robust with the respect to these choices.

We also display the change of basis matrix $M$ (as in \cref{algorithm:solving-eq-1}) that relates the torus-valued maps output by the two algorithms.

\subsection{Genus two surface}\label{subsec: genus2surface}

We apply \cref{construction:pipeline-examples} on a densely sampled surface of genus two (\cref{figure:example-coloring}), as in \cite[Section~3.9]{silva-morozov-vejdemo}.
As expected, persistent cohomology returns four high persistence features.
The resulting circular coordinates obtained by applying the Sparse Toroidal Coordinates Algorithm are shown in \cref{fig:gen2}~(Right).
For comparison, we show the circular coordinates obtained by applying the Sparse Circular Coordinates Algorithm to each cohomology class separately \cref{fig:gen2}~(Left).
The Dirichlet correlation matrices are as follows:
\[
D_{SCC} = \begin{pmatrix}
   3 & 2.4& 0 &-2.4\\
   2.4& 4.8& 0 &-4.8\\
   0 & 0 &22.6&11.2\\
  -2.4&-4.8&11.2&19.5
\end{pmatrix}\,,\;\;\;
D_{STC} = \begin{pmatrix}
 3 & -0.6 & 0 &  0 \\
-0.6 &  3 &   0 &   0 \\
 0 &   0 &  14.9 &  3.5\\
 0 &   0 &   3.5 & 14.8
\end{pmatrix} .
\]

\subsection{Lederman--Talmon dataset}
\label{section:dog-yoda}
We run \cref{construction:pipeline-examples} on a dataset collected and studied by Lederman and Talmon in~\cite{LEDERMAN2018509}.
In this example, two figurines \emph{Yoda} (the green figure on the left) and \emph{Dog} (the bulldog figure on the right) are situated on rotating platforms; see Figure~\ref{fig: dogyoda}.

Since each image is characterized by a rotation $(\phi_1,\phi_2)$ of the two figurines, we interpret the time series of images as an observation of a dynamical system on a two-torus.
Because the frequencies of rotation of both figurines have a large least-common-multiple, we expect the set of toroidal angles $(\phi_1,\phi_2)$  to comprise a dense sample of the torus, and verify this by treating the temporal sequence of images 
as a vector-valued time series and compute its sliding window persistence with window length $d=4$ and time delay $\tau=1$
(see Appendix~\ref{section:slidingwindow}).  


\begin{figure}[H]
    \centering
    \subfloat{\includegraphics[width=0.25\textwidth]{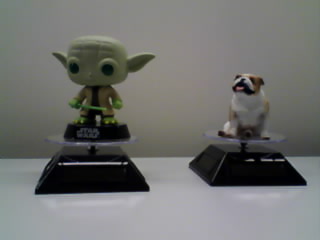}}
    \hfill
    \subfloat{\includegraphics[width=0.25\textwidth]{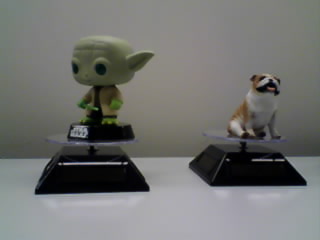}}
    \hfill
    \subfloat{\includegraphics[width=0.25\textwidth]{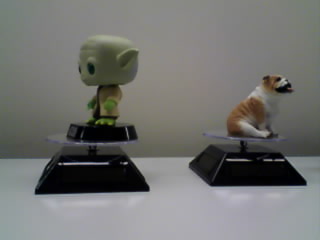}}
    \hfill
    \raisebox{-0.2cm}{\subfloat{\includegraphics[width=0.2\textwidth]{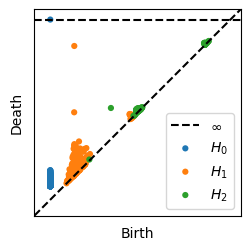}}}
    \caption{\textit{Left:} A sample of different images in the dataset.
    \textit{Right:} The sliding window persistence diagram of the data, showing two prominent 1-dimensional cohomology classes. \emph{Yoda}'s platform rotates clockwise completing about 310 cycles during the experiment, while in the same time \emph{Dog}'s platform completes about 450 cycles rotating counterclockwise. 
The data we consider are a collection of images of these rotating platforms captured from a fixed viewpoint.}
    \label{fig: dogyoda}
\end{figure}

In Figures~\ref{fig: yodadag_CC} and~\ref{fig: yodadag_CCnew}, we display the result of applying the Sparse Circular Coordinates Algorithm and the Sparse Toroidal Coordinates Algorithm to the sliding window point cloud of the dataset, respectively. 
Here, we show a sample of images as parameterized by the toroidal coordinates obtained from both algorithms.
%
The Dirichlet correlation matrices and change of basis matrix are as follows:
\[ D_{SCC} =
\begin{pmatrix}
 3.07 & -3.08\\
-3.08 & 10.48
\end{pmatrix}
\,,\;\;
D_{STC} = 
\begin{pmatrix}
3.07& 0\\
0    & 7.39
\end{pmatrix}
\,,\;\;
M=
\begin{pmatrix}
1 & 0\\
1    & 1
\end{pmatrix}.
\]

\begin{figure}[H]
    \centering
    \includegraphics[width=0.93\textwidth]{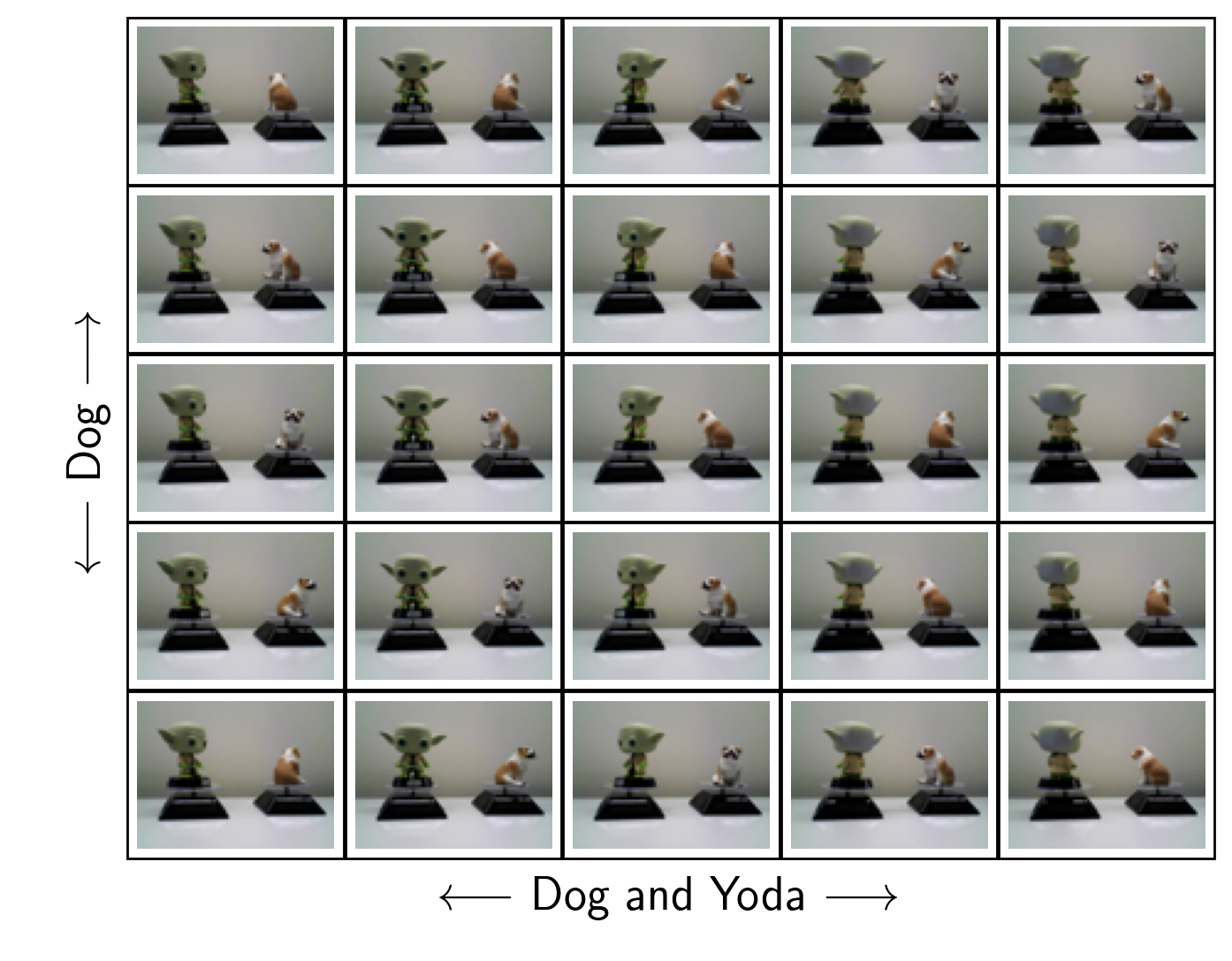}
    \caption{The vertical coordinate parameterizes \textit{Dog}'s rotation, while the horizontal coordinate parameterizes the rotation of both figurines.}
    \label{fig: yodadag_CC}
\end{figure}
\begin{figure}[H]
    \centering
    \includegraphics[width=0.93\textwidth]{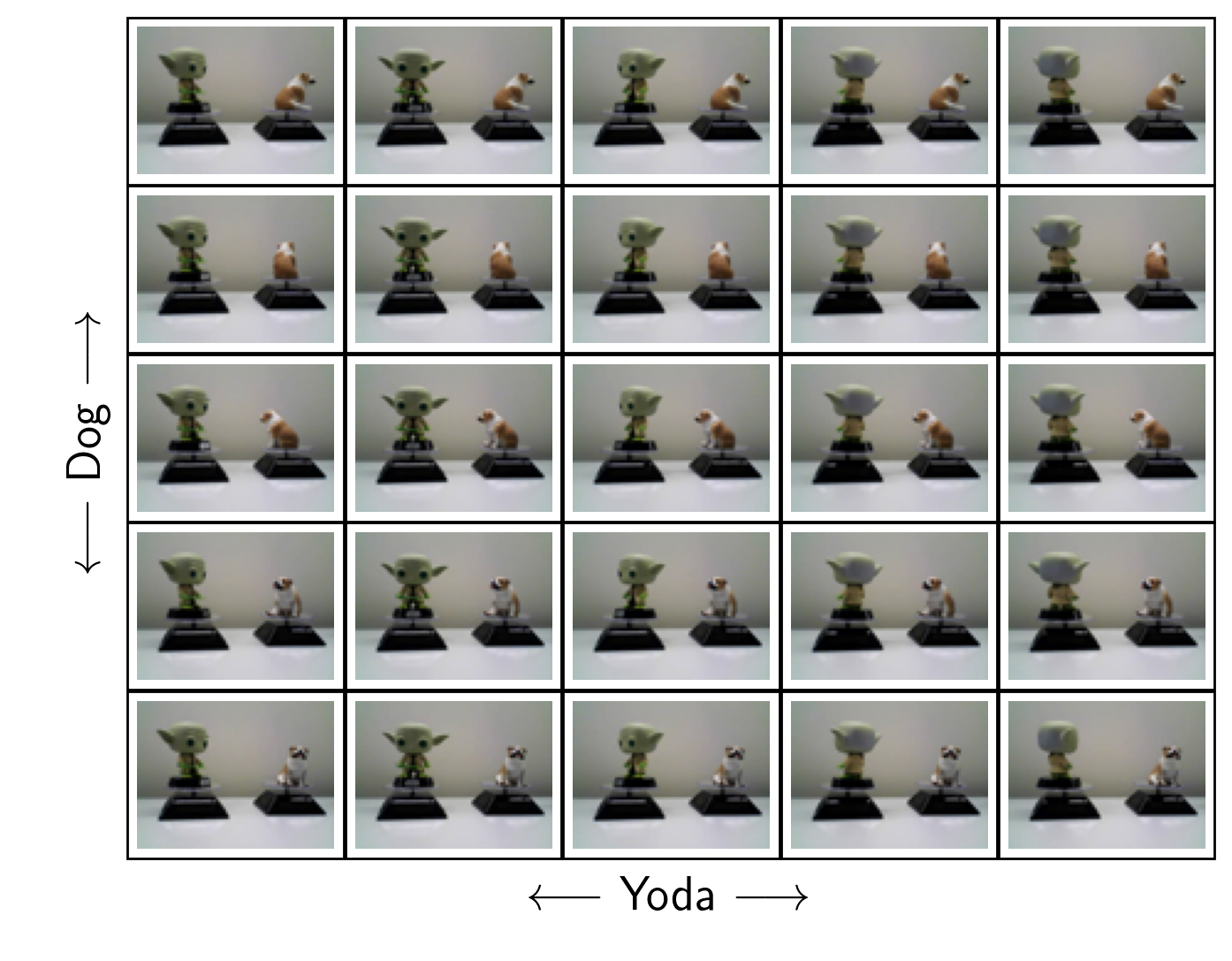}
    \caption{Vertical coordinate parameterizes \textit{Dog}'s rotation; horizontal parameterizes \textit{Yoda}'s.}
    \label{fig: yodadag_CCnew}
\end{figure}

\subsection{Kuramoto--Sivashinsky Dynamical Systems}
\label{section:KS}
An example of a one-dimensional Kuramoto--Sivashinsky (KS) equation is the following fourth order partial differential equation:
\[\frac{\partial u(x,t)}{\partial t} + 4 \;\frac{\partial^4u (x,t)}{\partial x^4} + 53.3\;\frac{\partial^2 u(x,t)}{\partial x^2} + 53.3\;u(x,t)\;\frac{\partial u(x,t)}{\partial x} = 0\]
with periodic boundary conditions $u(x,0) = \sin(x)$ and $u(0,t) = u(2\pi, t)$. 
The general family of KS equations \cite{hyman1986kuramoto} have gained popularity from their simple appearance and their ability to produce chaotic spatiotemporal dynamics. 
They have been shown to model pattern formations in several physical contexts; for instance \cite{kuramoto1976persistent, michelson1977nonlinear, sivashinsky1980irregular}.

The underlying dynamical system is toroidal. Indeed, it is controlled by two frequencies: one comes from oscillation in time, the other is dictated by the speed of the traveling wave along the periodic domain. 
However, the dynamic is not periodic and the trajectory of any initial state eventually densely fills out the torus. 
We represent the solution $u(x,t)$ to this equation as a heatmap in \cref{fig: KS} (Left).
The horizontal axis refers to time $t$, the vertical axis refers to the spatial variable $x$. At each time, $u(x,t)$ is periodic and a slice of it can be seen in \cref{fig: KS} (Middle).
We treat $u(x,t)$ as a vector valued time series $f(t): = u(-,t)$ and compute the sliding window persistence
(\cref{section:slidingwindow})
of $f$ with parameters $d = 5$ and $\tau =4$; see \cref{fig: KS} (Right) for the resulting persistence diagram.

\begin{figure}[H]
    \centering
    \includegraphics[width=0.3\textwidth]{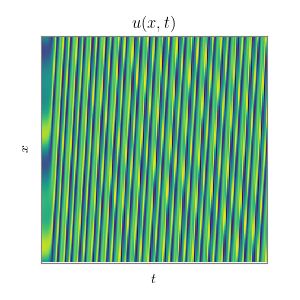}
    \hfill
    \includegraphics[width=0.3\textwidth]{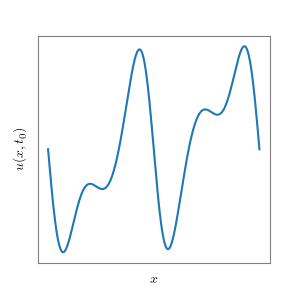}\hfill
    \raisebox{-0.45cm}{\includegraphics[width=0.32\textwidth]{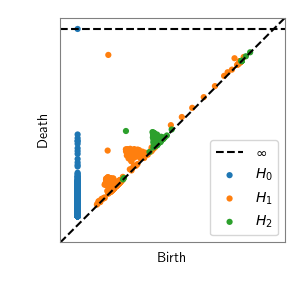}}
    \caption{(Left) The solution to the KS equation as a heatmap; (Middle) A slice of $u(x,t)$ at a fixed $t_0$; (Right) The sliding window persistence digaram.}
    \label{fig: KS}
\end{figure}


In Figures~\ref{fig: KSPM_old} and~\ref{fig: KSPM_new}, we display the result of applying the Sparse Circular Coordinates Algorithm and the Sparse Toroidal Coordinates Algorithm to the sliding window embedding of the dataset, respectively. 
As in Example~\ref{section:dog-yoda}, we show sample data points (in this case waves) as parameterized by the toroidal coordinates obtained from both algorithms.

To verify that the vertical and horizontal components of \cref{fig: KSPM_new} are indeed parameterizing oscillation and traveling, respectively, we partition the dataset in 50 bins, according to the vertical coordinate, and in each bin we compute all pairwise rotationally invariant $L^2$ distances between the waves.
We show the histogram of distances in \cref{figure:density-ks}.
The Dirichlet correlation matrices and change of basis are as follows:
\[ D_{SCC} =
\begin{pmatrix}
1.1&   6.1\\
6.1& 112.3
\end{pmatrix}
\,,\;\;
 D_{STC}=
\begin{pmatrix}
1.1 &  1\\
1 & 76.6
\end{pmatrix}
\,,\;\;
M=
\begin{pmatrix}
1 &  0\\
-5 & 1
\end{pmatrix}.
\]


\begin{figure}[H]
\centering
    \includegraphics[width=.85\textwidth]{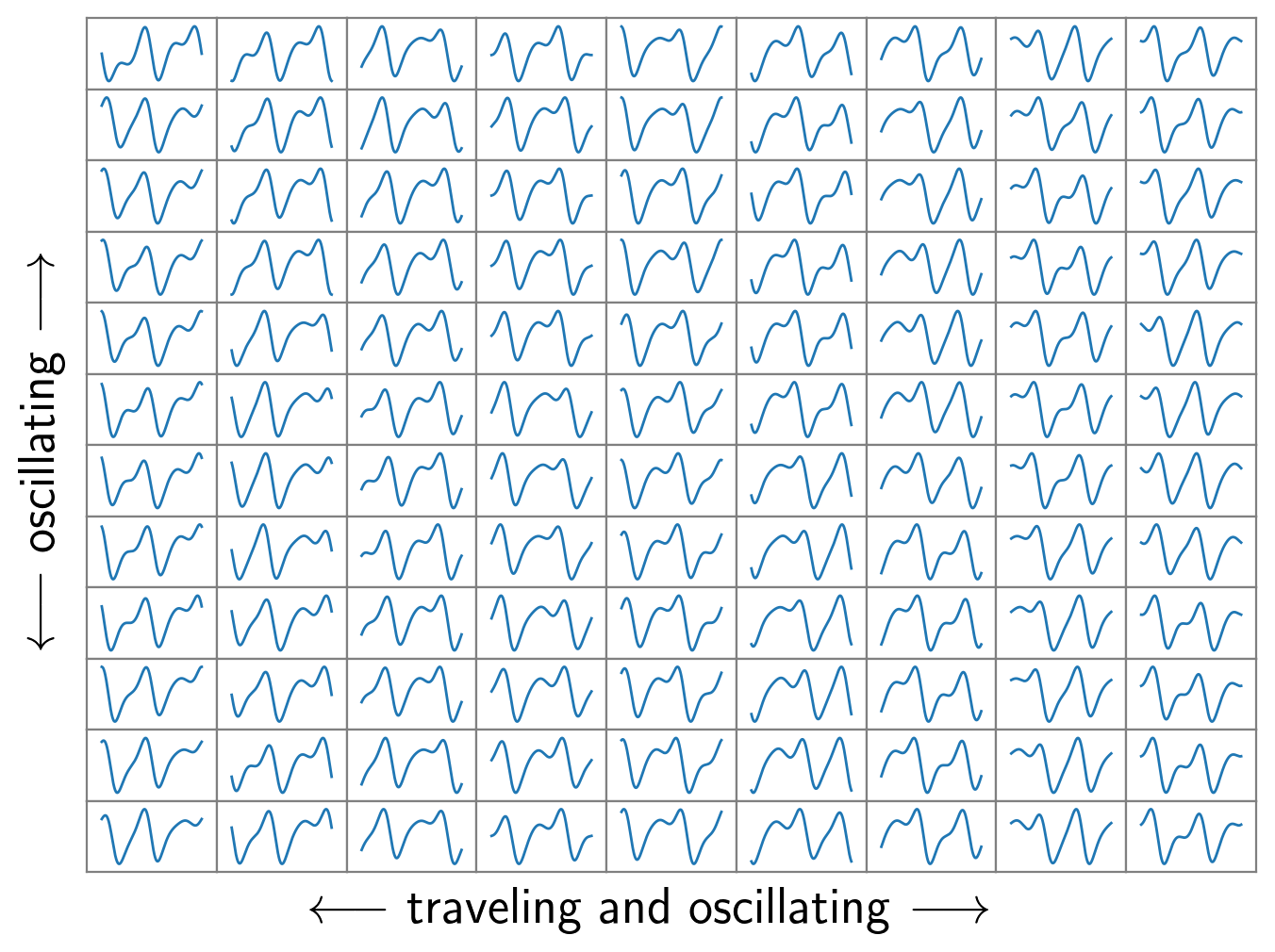}
    \caption{Oscillatory behavior vertically and combination of traveling and oscillatory horizontally.}
    \label{fig: KSPM_old}
\end{figure}
\begin{figure}
    \centering
    \includegraphics[width=.85\textwidth]{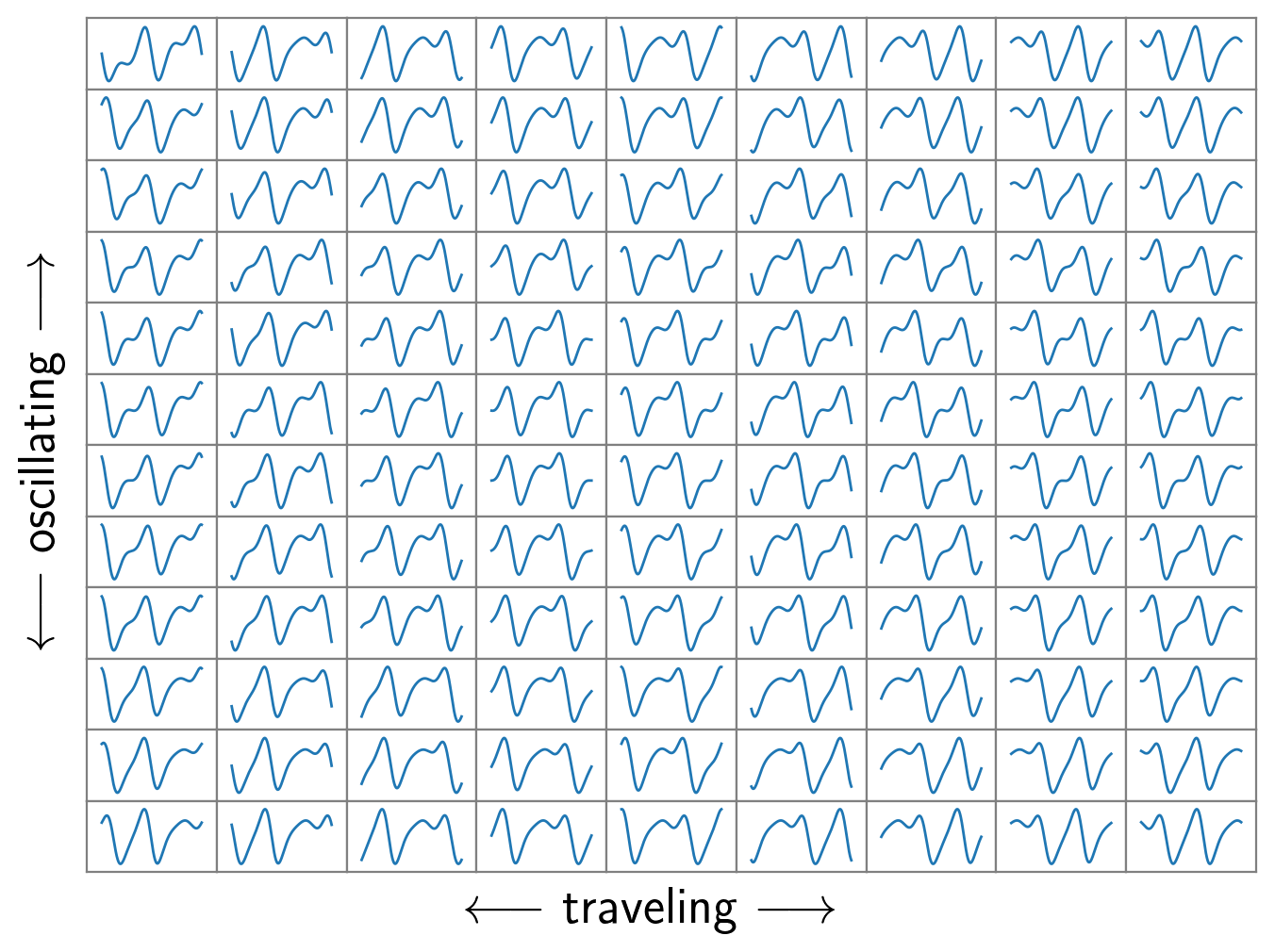}
    \caption{Traveling waves parameterized horizontally and oscillations parameterized vertically.}
    \label{fig: KSPM_new}
\end{figure}

\begin{figure}
    \centering
    \includegraphics[width=0.6\textwidth]{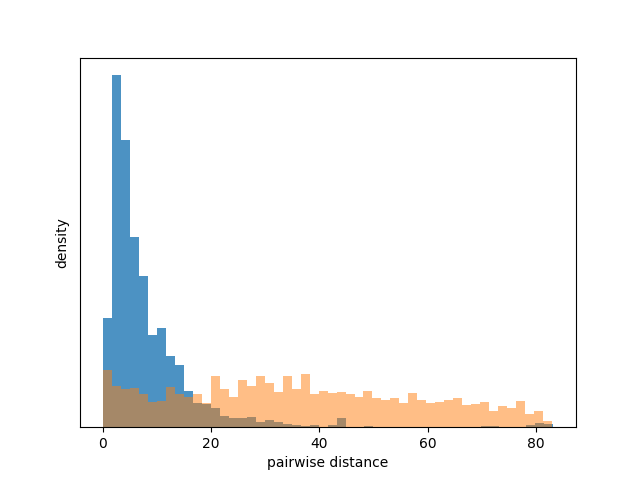}
    \caption{Density of pairwise distances using the parameterization of the Toroidal Coordinates Algorithm is in blue, while density of pairwise distances using the parameterization given by the Circular Coordinates Algorithm is in red.
    This suggests that, indeed, the horizontal coordinate of \cref{fig: KSPM_new} parameterizes rotation only while the one of \cref{fig: KSPM_old} does not.}
    \label{figure:density-ks}
\end{figure}


\subsection{Synthetic Neuroscience Example}
\label{section:neuroscience}
We show that our methods are a viable way of constructing informative circle-valued representations of neuroscientific data.
Place cells in the mammalian hippocampus have spatially localized receptive fields that encode position by firing rapidly at specific locations as one navigates an environment \cite{place_cells}.
It is shown in \cite{Finkelstein_three_dimensional_coding} that head direction of bats is encoded in a similar way: certain neurons are tuned to pitch, others to azimuth, and others to roll, each using a circular coordinate system to do so.

We consider a synthetic dataset inspired by these types of neuronal responses to stimuli. 
Suppose three populations of neurons $P_1$, $P_2$, and $P_3$ are tuned to elevation, azimuth, and roll, respectively.
This means, for instance, that if neuron $n \in P_1$ is tuned to a head elevation of $45$ degrees, then $n$ fires most rapidly when the head is at an elevation of $45$ degrees, fires less rapidly if the head is at an elevation of, say $35$ or $55$ degrees, and maintains low activity near an elevation of $0$ or 90 degrees.
Now, suppose we record the firing rates of neurons in populations $P_1,P_2,$ and $P_3$ for a duration of $T$ time steps while an animal moves its head freely.
Letting $N$ denote the total number of recorded neurons, we can record this data as an $N \times T$ matrix $M$, where $M_{i,j}$ corresponds to the firing rate of neuron $i$ at time step $j$.

We interpret $M$ as a collection of $T$ points in $\mathbb{R}^n$ and consider the problem of recovering three circular coordinates $M \to \circle$ (one each for elevation, azimuth, and roll) that map a point of $M$, thought of as a time step, to the correct head orientation at that time.
We construct a synthetic dataset simulating the situation above
(\cref{section:neuro-data})
and run it through \cref{construction:pipeline-examples}.
\cref{fig:neuro_persistence} shows the resulting persistence barcode, which contains three prominent $1$-dimensional cohomology classes.

We then use the Sparse Circular Coordinates Algorithm to recover a circle-valued map from each of the three most prominent 1-dimensional cohomology classes. 
To determine whether we successfully recover all three orientations, we use a scatter plot as in \cite[Section~3.2]{silva-morozov-vejdemo}, and display the three recovered maps $M \to \circle$ against the ground truth.
The Sparse Circular Coordinates Algorithm run independently on each cohomology class fails to recover the three head orientations (Figure~\ref{fig_neuro_main_fig}, left) in many of our runs, while the Toroidal Coordinates Algorithm always recovers the three coordinates (\cref{fig_neuro_main_fig}, right).
The Dirichlet correlation matrices and the change of basis are as follows:
\[ D_{SCC}=
\begin{pmatrix}
 51.6& -0.9&  2.7\\
 -0.9& 50.2&-99.7\\
  2.7&-99.7&249.8
\end{pmatrix}
\,,\;\;
D_{STC}=
\begin{pmatrix}
51.6& -0.9&  0.8\\
-0.9& 50.2&  0.5\\
 0.8&  0.5& 51.1
\end{pmatrix}
\,,\;\;
M=
\begin{pmatrix}
1& 0& 0\\
0& 1& 0\\
0& 2& 1
\end{pmatrix}.
\]

\begin{figure}[h]
\centering
\includegraphics[width=.5\linewidth]{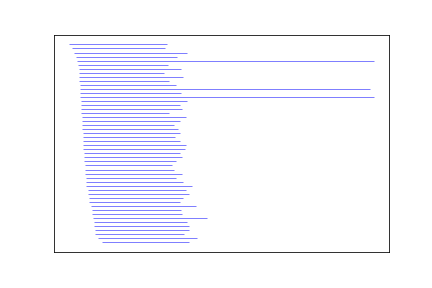}
\vspace*{-0.4cm}
\caption{Persistence barcode of exhibiting three prominent 1-dimensional cohomology classes.}
\label{fig:neuro_persistence}
\end{figure}

\begin{figure}[H]

    \begin{subfigure}[b]{0.45\linewidth}
    \centering
    Sparse Circular Coordinates
    
    \includegraphics[scale=0.27]{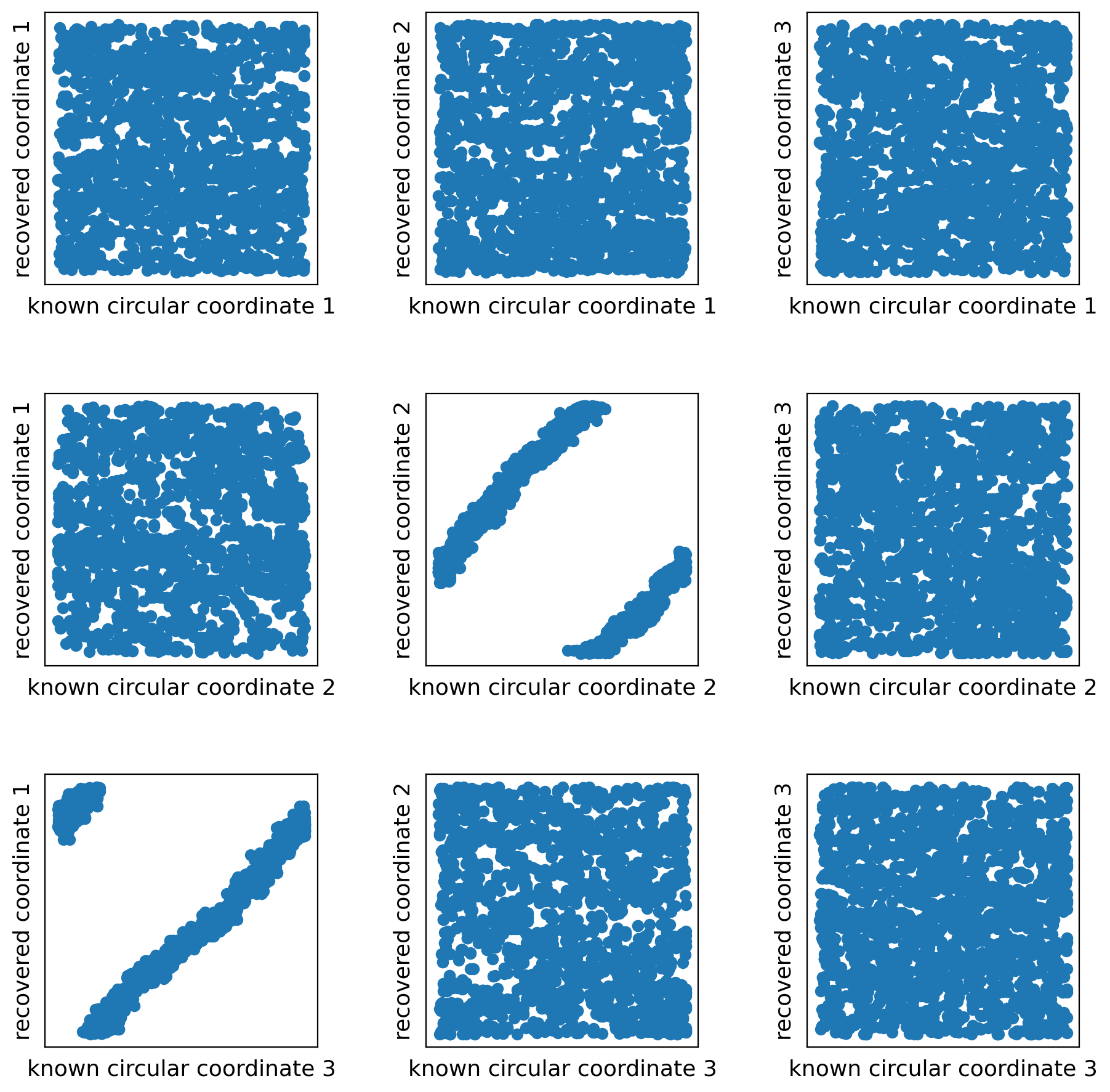}\end{subfigure}
    \hfill
    \begin{subfigure}[b]{0.45\linewidth}
    \centering
    Sparse Toroidal Coordinates
    
    \includegraphics[scale=0.27]{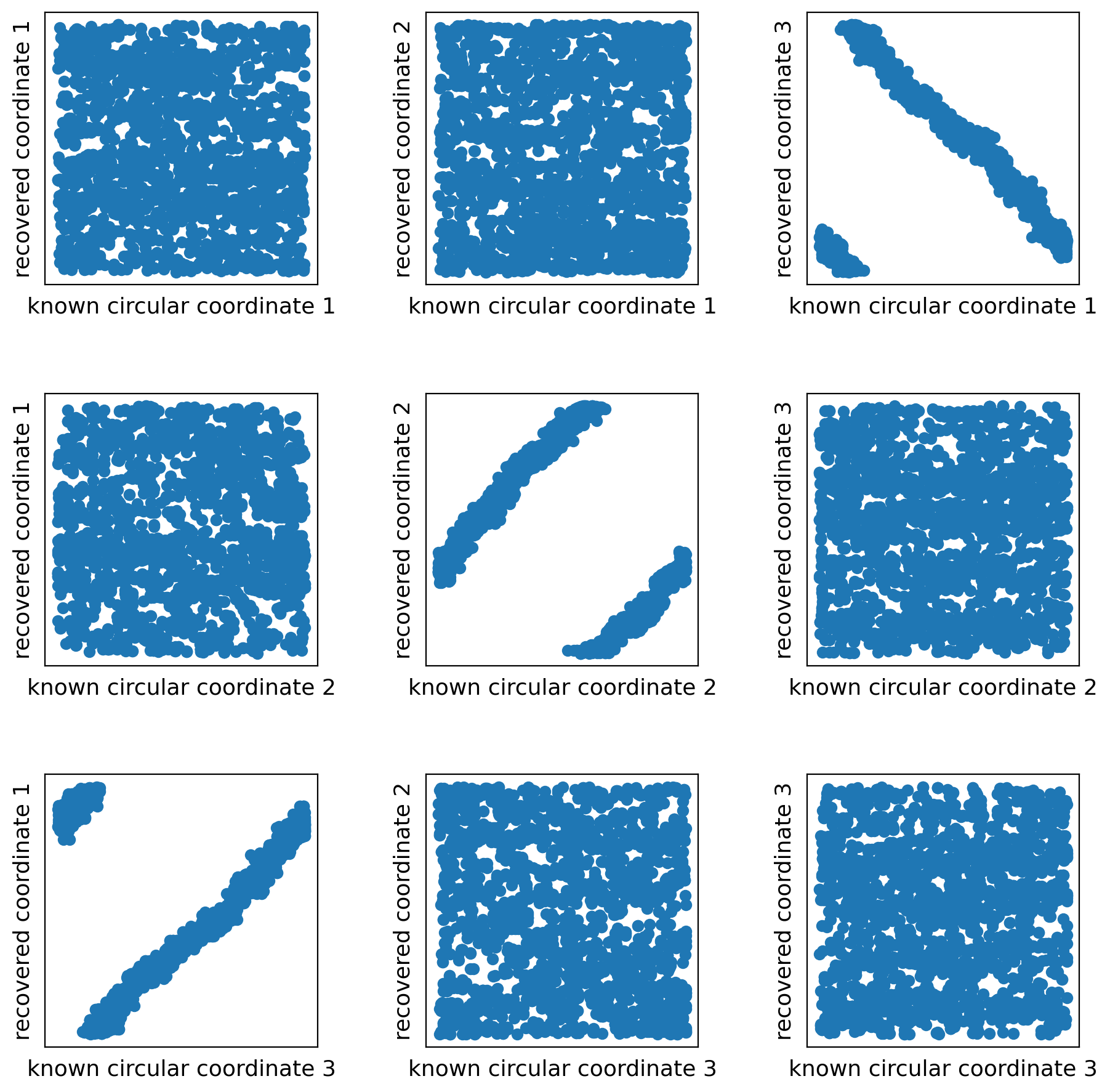}
    \end{subfigure}
    
    \caption{Recovered versus known circular coordinates using the Sparse Circular Coordinates Algorithm and the Toroidal Coordinates Algorithm.}
    \label{fig_neuro_main_fig}
\end{figure}


 

\newpage

\bibliography{bibliography}

\appendix

\newpage
\section{Proofs}

\subsection{Correctness of lattice reduction procedure}
\label{section:proofs-lattice-reduction}

We fix a finite simplicial complex $K$, cohomology classes $\alpha_1, \dots, \alpha_k \in \simphom(K;\Zbb)$, and an inner product $\innerprod$ on $\simpco(K;\Rbb)$.

We start with a definition.
Recall from \cref{section:background} that we have a linear map $\delta : \simpchzero(K;\Rbb) \to \simpco(K;\Rbb)$ and that $\simphom(K;\Rbb) = \simpco(K;\Rbb)/\Im(\delta)$.
Using the inner product on $\simpco(K;\Rbb)$, we get the orthogonal projection linear map $\proj_{\Im(\delta)} : \simpco(K;\Rbb) \to \Im(\delta)$.

\begin{definition}
    A cocycle $\theta \in \simpco(K;\Rbb)$ is \define{harmonic} with respect to the inner product $\innerprod$ if it satisfies
    $\proj_{\Im(\delta)}(\theta) = 0$.
    We say that $\theta$ is a \define{harmonic representative} of a cohomology class $\alpha \in \simphom(K;\Rbb)$ if $\theta$ is harmonic and $[\theta] = \alpha$.
    The subspace of harmonic cocycles is denoted by $\harco(K) \subseteq \simpco(K;\Rbb)$.
\end{definition}

Note that, by definition, we have $\harco(K) = \Im(\delta)^\bot$, and the linear map $h : \harco(K) \to \simphom(K;\Rbb)$ given by the following composite
\begin{equation}
    \label{equation:hodge-isomorphism}
    \harco(K) = \Im(\delta)^\bot \xhookrightarrow \Im(\delta)^\bot \oplus \Im(\delta) = \simpco(K;\Rbb) \twoheadrightarrow \simpco(K;\Rbb)/\Im(\delta) = \simphom(K;\Rbb)
\end{equation}
is an isomorphism of real vector spaces.
Let $h^{-1} : \simphom(K;\Rbb) \to \harco(K)$ denote the inverse linear isomorphism.

\begin{lemma}
    \label{lemma:unique-harmonic-representative}
    Any cohomology class $\alpha \in \simphom(K;\Rbb)$ admits a unique harmonic representative.
    Moreover, a cocycle $\theta \in \simpco(K;\Rbb)$ is harmonic if and only if it is a solution of $\argmin \left\{ \|\eta\|^2 \;|\; [\eta] = [\theta] \in \simphom(K;\Rbb)\right\}$.
\end{lemma}
\begin{proof}
    The first statement follows from the fact that $h$ is an isomorphism.
    
    For the second statement, note that any cocycle $\eta \in \simpco(K;\Rbb)$ can be written in a unique way as $\eta_1+ \eta_2$ with $\eta_1 \in \Im(\delta)^\bot$ and $\eta_2 \in \Im(\delta)$, and this is such that $\|\eta_1\|^2 = \|\eta_1\|^2 + \|\eta_2\|^2$.
    Given $\theta \in \simpco(K;\Rbb)$, with the above notation, and by the isomorphism of \cref{equation:hodge-isomorphism}, we have $[\eta] = [\theta] \in \simphom(K;\Rbb)$ if and only if $\eta_1 = \theta_1 \in \Im(\delta)^\bot$.
    Thus, the unique $\eta$ minimizing $\min \left\{ \|\eta\|^2 \;|\; [\eta] = [\theta] \in \simphom(K;\Rbb)\right\}$ is $\eta = \theta_1$.
    And, by definition, we have $\theta = \theta_1$ if and only if $\theta$ is harmonic.
\end{proof}

In particular, $h^{-1}(\alpha)$ provides a solution to \cref{problem:minimization-circular-coordinates}.
The following is clear.

\begin{corollary}
    \label{corollary:solution-is-harmonic}
    If $\theta_1, \dots, \theta_k \in \simpco(K;\Rbb)$ is a solution to \cref{equation:minimization-toroidal-coordinates}, then $\theta_j$ is harmonic for all $1 \leq j \leq k$.\qed
\end{corollary}

Let $R_\alpha \subseteq \harco(K)$ be the real vector space spanned by $\{h^{-1}(\alpha_1), \dots, h^{-1}(\alpha_k)\} \subseteq \harco(K)$.
Endow $R_\alpha$ with the inner product inherited from the inner product of $\harco(K) \subseteq \simpco(K)$.
Let $L_\alpha \subseteq R_\alpha$ be the lattice generated by taking all integer linear combinations of $\{h^{-1}(\alpha_1), \dots, h^{-1}(\alpha_k)\}$.
Note that this is indeed a lattice (does not have accumulation points) since $\{\alpha_1, \dots, \alpha_k\}$ is linearly independent.

\begin{lemma}
    \label{proposition:reduction-eq1-to-lattice-reduction}
    A set of cocycles $\theta_1, \dots, \theta_k \in \simpco(K;\Rbb)$ is a solution to \cref{equation:minimization-toroidal-coordinates} 
    if and only if $\{\theta_1, \dots, \theta_k\}\subseteq L_\alpha$ and 
    $\theta_1, \dots, \theta_k$ is a solution to \cref{problem:our-lattice-reduction} with lattice $L_\alpha \subseteq R_\alpha$.
\end{lemma}
\begin{proof}
    The result follows at once from the following observation.
If $\theta_1, \dots, \theta_k \in \simpco(K;\Rbb)$ is a solution to \cref{equation:minimization-toroidal-coordinates}, then, by \cref{corollary:solution-is-harmonic}, the cocycle $\theta_j$ is harmonic for all $j$; and moreover, since by assumption the sets $\{[\theta_j]\}_{1 \leq j \leq k}$ and $\{\alpha_j\}_{1 \leq j \leq k}$ generate the same Abelian subgroup of $\simphom(K;\Rbb)$, we have that $\{\theta_j\}_{1 \leq j \leq k}$ must form a basis of $L_\alpha$.
\end{proof}

\begin{proof}[Proof of \cref{proposition:LLL-gives-approximate-solution}]
    Note that all the arguments in \cite{lenstra-lenstra-lovasz} up to \cite[Proposition~1.26]{lenstra-lenstra-lovasz} are done for arbitrary lattices $L \subseteq \Rbb^n$.
    In the proof of \cite[Proposition~1.12]{lenstra-lenstra-lovasz}, it is shown, in particular, that given any linearly independent set $X = \{x_1, \dots, x_t\} \subseteq L$, there exists a reordering of $X$, say $\{x'_1, \dots, x'_t\}$, such that $\|b_i\|^2 \leq 2^{k-1} \|x'_j\|^2$ for all $1 \leq j \leq t$.
    This proves the claim.
\end{proof}

\begin{proof}[Proof of \cref{proposition:algo-eq-1-correct}]
    Note that \cref{proposition:reduction-eq1-to-lattice-reduction} implies that finding an approximate solution to \cref{equation:minimization-toroidal-coordinates} is equivalent to finding an approximate solution to \cref{problem:our-lattice-reduction} with the lattice generated by $\{\eta_1, \dots, \eta_k\}$, which in turn is equivalent to finding an approximate solution to \cref{problem:our-lattice-reduction} with the lattice generated by the rows of $C$, by definition of the Cholesky decomposition.
    Finally, the LLL-algorithm provides such an approximate solution by \cref{proposition:LLL-gives-approximate-solution}.
\end{proof}

\subsection{Proof of \cref{proposition:analogy-toroidal-coordinates}}
\label{section:proof-of-analogy-toroidal-coords}

\begin{lemma}
    \label{lemma:energy-is-sum-of-energies}
    Let $\Mcal$ be a closed Riemannian manifold and let $k \in \Nbb$.
    Let $f = (f_1, \dots, f_k) : \Mcal \to \torus^k$ be a smooth map with $f_i : \Mcal \to \circle$ for $1 \leq i \leq k$.
    Then $E[f] = \sum_{i=1}^k E[f_i]$.
\end{lemma}
\begin{proof}
    This follows at once from the definition of Dirichlet energy and the fact that $\torus^k = \circle \times \cdots \times \circle$ is endowed with the product Riemannian metric.
%
\end{proof}

%

\begin{lemma}
    \label{lemma:differential-of-circlular-coordinates}
    The function $f : \Mcal \to \circle$ obtained using \cref{construction:continuous-integration} on a $1$-form $\theta$ is smooth and satisfies $df = \theta$.
\end{lemma}
\begin{proof}
    Without loss of generality, we may assume that $\Mcal$ is connected.
    Note that, if $p,q : [0,1] \to \circle$ are two smooth paths between points $x,y \in \Mcal$, then 
    $\left(\int_0^1 \theta_{p(t)}(p'(t)) dt\right) \mod \mathbb{Z} = \left(\int_0^1 \theta_{q(t)}(q'(t)) dt\right) \mod \mathbb{Z}$.
    This is because the concatenation of $p$ and the inverse path of $q$ is a closed loop and any closed $1$-form representing an integral class integrates to an integer on any closed loop.
    This last fact can be seen, for instance, by recalling that the isomorphism between de Rham and singular cohomology is given by integration on chains (see, e.g., \cite[Chapter~1]{dupont}).
    This shows that the definition of $f$ is independent of the choices of paths $p$.
    In particular, we have
    \begin{equation}
        \label{equation:f-independent-path}
        f(y_2) = f(y_1) + \left(\int_0^1 \theta_{p(t)}(p'(t)) dt\right) \mod \mathbb{Z}
    \end{equation}
    for any smooth path $p : [0,1] \to \Mcal$ between $y_1$ and $y_2$.
    Now, if $y_1,y_2 \in \Mcal$ are sufficiently close, then there exists $y \in \Mcal$ and a smooth family of paths starting with a path between $y$ and $y_1$ and ending with a path between $y$ and $y_2$.
    Since definition of $f$ is independent of the chosen paths, this shows that $f$ is smooth.

    To conclude, let $v \in \tangent_y \Mcal$ and let $p : [-1,1] \to \Mcal$ such that $p(0) = y$ and $p'(0) = v$.
    Let $y_1 = p(-1)$ and $y_2 = p(1)$.
    From \cref{equation:f-independent-path} now follows that $df_y(v) = \theta_y(v)$.
    Thus, $df = \theta$, as required.
\end{proof}

\begin{proof}[Proof of \cref{proposition:analogy-toroidal-coordinates}]
    Without loss of generality, we may assume that $\Mcal$ is connected.
 Then, \cref{lemma:differential-of-circlular-coordinates} implies that there is a bijection between the set of smooth maps $f : \Mcal \to \circle$ up to rotational equivalence on one hand, and the set of closed $1$-forms $\theta \in \Omega^1(\Mcal)$ with $[\theta]$ in the image of $\iota : \simphom(\Mcal;\Zbb) \to \simphom(\Mcal;\Rbb)$ on the other hand.
    Under this correspondence, we have $E(f) = \frac{1}{2} \|\theta\|^2$, by definition.

The correspondence extends to a bijection between the set of smooth maps $f : \Mcal \to \torus^k$ up to composition with a component-wise rotation of $\torus^k = \circle \times \cdots \times \circle$ on one hand, and the set of ordered lists of $k$ closed $1$-forms $\{\theta_1, \dots, \theta_k\} \subseteq \Omega^1(\Mcal)$ with $[\theta_j]$ in the image of $\iota$ for all $1 \leq j \leq k$, on the other hand.
    Under this correspondence, we have $E(f) = \frac{1}{2} \sum_{j = 1}^k \|\theta_j\|^2$, by \cref{lemma:energy-is-sum-of-energies}.
    The result follows.
\end{proof}

\subsection{Proof of \cref{proposition:isometry}}
\label{section:proof-of-isometry}

Let $f,g : \Mcal \to \circle$ be obtained using sparse cocycle integration (\cref{algorithm:sparse-cocycle-integration}) with input cocycles $\theta$ and $\eta$, respectively.
We have
\begin{align*}
    D(f,g) &= \frac{1}{2}\; \langle df, dg \rangle_{\Omega^1}\\
    &= \frac{1}{2}\;\int_{b \in \Mcal}  \langle df_b, dg_b \rangle_F \;\dsf \mu(b)\\
    &=  \frac{1}{2}\; \sum_{w \in I}\int_{b \in \Mcal} \langle df_b, dg_b \rangle_F \; \phi_w(b) \;\dsf \mu(b)\\
    &=  \frac{1}{2}\; \sum_{w \in I} \left( \int_{b \in \Mcal} \left\langle \sum_{y \in I} d(\phi_y)_b \theta^{wy}\;,\; \sum_{z \in I} d(\phi_z)_b \eta^{wz}\right\rangle_F \; \phi_w(b) \;\dsf \mu(b)\right)\\
    &=  \frac{1}{2}\; \sum_{w,y,z \in I} \left( \int_{b \in \Mcal} \;\langle d(\phi_y)_b, d(\phi_z)_b \rangle_F\; \phi_w(b)\; \dsf\mu(b)\right) \theta^{wy} \eta^{wz}\\
    &= \frac{1}{2}\; \sum_{w,y,z \in I} D_{wyz} \; \theta^{wy} \eta^{wz},
\end{align*}
as required.

\section{Estimating the Dirichlet form of arbitrary circle-valued maps}\label{section:estimatingdirichlet}
We give a heuristic for estimating the Dirichlet form between arbitrary circle-valued maps on a Riemannian manifold.
Formally addressing the consistency of this heuristic is left for future work.

\begin{construction}
    \label{construction:heuristic-dirichlet-form}
    Let $X \subseteq \Rbb^n$ be a finite sample of a smoothly embedded closed manifold $\Mcal \subseteq \Rbb^n$.
    Given $f,g : X \to \circle$, restrictions of smooth maps $\tilde{f},\tilde{g} : \Mcal \to \circle$, we seek to estimate $D(\tilde{f},\tilde{g})$.
    \begin{enumerate}
        \item Form a neighborhood graph $G$ on $X$.
        For instance, this can be done by selecting $k \in \Nbb$ and using an undirected $k$-nearest neighbor graph.
        \item Compute weights $h(a,b) \geq 0$ for the edges $(a,b) \in G$.
        For instance, this can be done by selecting a radius $\delta > 0$ and letting $h(a,b) = \exp(-\|a-b\|^2/\delta^2)$.
        \item Note that $\Rbb \to \circle$ restricts to a bijection $[-1/2, 1/2) \to \circle$ and let $l :  \circle \to [-1/2, 1/2)$ denote its inverse.
        \item For $a \in X$, let $N(a) = \{b \in G \mid (a,b) \in G\}$, and define 
        \[
            \estdirichlet(f,g) \coloneqq \sum_{a \in G}  \left(  \frac{1}{N(a)}\sum_{b \in N(a)} h(a,b) \; l(f(b) - f(a)) \; l(g(b) - g(a)) \right).
        \]
    \end{enumerate}
\end{construction}

We next define the notion of Dirichlet correlation matrix that we use to quantify the correlation between circle-valued map in the examples.

\begin{definition}
Given $f_1, \dots, f_k : X \to \circle$, we define its \define{Dirichlet correlation matrix} $D(f_1, \dots, f_k)$ as the matrix with entry $(i,j)$ given by $\estdirichlet(f_i,f_j)$.
\end{definition}

\section{Details about examples}

\subsection{Sliding Window Persistence}\label{section:slidingwindow}
Given a vector-valued function $F\colon \mathbb{R}\to\mathbb{R}^N$ and parameters $d,\tau\in\mathbb{N}$, the sliding window embedding $SW_{d,\tau}F\colon \mathbb{R}\to\mathbb{R}^{N\times(d+1)}$ of $F$ is defined as follows:
\begin{equation}\label{equation:slidingwindow}
    SW_{d,\tau}F(t) = \begin{bmatrix}
    F(t) & F(t+\tau) & F(t+2\tau) & \cdots & F(t+d\tau)
\end{bmatrix}^T
\end{equation}
where $d \in \Nbb$ is the embedding dimension and $\tau \in \Nbb$ is the time delay. 
If $F$ is assumed to be an observation of a dynamical system,  then for appropriate choices of parameters $d,\tau$, the collection $\mathbb{SW}_{d,\tau}F\coloneqq \{SW_{d,\tau}F(t) \; | \; t \in \Rbb\}$ is topologically equivalent to the observed trajectory. This is a consequence of Takens' theorem \cite{takens1981detecting}. The $1$-sliding window persistence is the persistent (co)homology of the Vietoris-Rips filtration of a finite $L \subset \mathbb{SW}_{d,\tau}F$. See \cite{tralie2018quasi} for a treatment of sliding window embeddings of vector-valued functions.

 

\subsection{Construction of neuroscience data}
\label{section:neuro-data}

Our synthetic dataset is constructed as follows.
For $i=1,2,3$, 
let $C_i$ denote a circle of circumference $1$ on which we have placed $6$ uniformly distributed sensors that fire at a rate inversely proportional to the distance of some stimulus on $C_i$. 
On each circle, we take $50 = n_T$ random walks of $T = 50$ steps and record the sensor responses as an $N \times (n_T \times T)$ matrix.
For the sensor response, we use $r(d) = \max(0,1-3d)$, where $d$ is the circular distance from the sensor to the position of the walk at a given time step.

\end{document}